\documentclass[a4paper, 11pt]{article}

\usepackage[latin1]{inputenc}
\usepackage{amstext, amssymb, amsthm, amsfonts, latexsym,bbm}
\usepackage[tbtags]{amsmath}
\usepackage{hyperref}
\usepackage{color}
\usepackage{graphicx}
\usepackage{enumerate}
\usepackage[backgroundcolor=yellow]{todonotes}
\usepackage{caption, subcaption}
\usepackage{authblk}

\newcommand{\lp}{\left(}
\newcommand{\rp}{\right)}

\newcommand{\lac}{\left\{}
\newcommand{\rac}{\right\}}

\def\E{\mathbb{E}}
\def\R{\mathbb{R}}
\def\N{\mathbb{N}}

\def\Q{\mathbb{Q}}

\def\Prob{\mathbb{P}}
\newcommand{\ind}{\boldsymbol{1}}

\newtheorem{theorem}{Theorem}
\newtheorem{proposition}[theorem]{Proposition}

\newtheorem{lemma}[theorem]{Lemma}
\newtheorem{corollary}[theorem]{Corollary}

\theoremstyle{definition}

\newtheorem{example}[theorem]{Example}

\newtheorem{assump}[theorem]{Assumption}

\begin{document}

\title{Ultrasensitivity and sharp threshold theorems for multisite systems
}

\author[1]{Micha\"el Dougoud, Christian Mazza and Laura Vinckenbosch}
\affil[1]{University of Fribourg, Department of Mathematics, Chemin du Mus\'ee 23, CH-1700 Fribourg, Switzerland}
\date{\today}

\maketitle

\begin{abstract}
We study the ultrasensitivity of multisite binding processes where ligand molecules can bind to several binding sites, considering more particularly recent models involving complex chemical reactions in phosphorylation systems such as allosteric phosphorylation processes, or substrate-catalyst chain reactions and nucleosome mediated cooperativity. New statistics based formulas for the Hill coefficient and the effective Hill coefficient are provided and necessary conditions for a system to be ultrasensitive are exhibited. We first show that the ultrasensitivity of binding processes can be approached using sharp-threshold theorems which have been developed in applied probability theory and statistical mechanics for studying sharp threshold phenomena in reliability theory, random graph theory and percolation theory. We hence introduce influence functions and show that general results can be obtained for monotone measures.
We then assume  that the binding process is described by a density dependent birth and death process. We provide precise large deviation results for the steady state distribution of the process, and show that switch-like ultrasensitive responses are strongly related to the multi-stability of the associated dynamical system. Ultrasensitivity occurs if and only if the entropy of the dynamical system has more than one global minimum for some critical ligand concentration. In this case, the Hill coefficient is proportional to the number of binding sites, and the systems is highly ultrasensitive. We also discuss the interpretation of an extension $I_q$ of the effective Hill coefficient $I_{0.9}$ for which we recommend the computation of a broad range of values of $q$ instead of just the standard one corresponding to the 10\% to 90\% variation in the dose-response. It is shown that this single choice can sometimes mislead the conclusion by not detecting ultrasensitivity. This new approach allows a better understanding of multisite ultrasensitive systems and provides new tools for the design of such systems.
\end{abstract}

\section{Introduction}

Ultrasensitive responses, that is, switch-like input-output relationships are common-place in signal transduction systems involving signaling cascades or bistable switches, see, e.g., the review articles \cite{ferrell2014a,ferrell2014b,ferrell2014c}. We focus on switching mechanisms based on multisite phosphorylation processes, or, more generally, on multisite binding processes, where ligand molecules can bind cooperatively to $N$ binding sites. Such processes create thresholds such that the proportion of highly phosphorylated substrate is close to 0 when the ratio of kinase to phosphatase activity is below a critical level. The system is ultrasensitive if the response switches abruptly from a low to a high phosphorylation level when the ratio of kinase to phosphatase crosses the critical threshold. Usually this occurs when $N$ is large, but the fact of having many phosphorylation sites is not sufficient to ensure ultrasensitivity, see \cite{Gunawardena2005} or \cite{Thomson2009}. Various processes like protein or enzyme sequestration \cite{Buchler2009,Martins2013} or allosteric mechanisms, see e.g., \cite{tenWolde,enciso2013,Martins2013,EncisoKelloggVargas2014,Edelstein} are known to induce ultrasensitivity. Our approach considers stochastic kinetics based on birth and death processes modeling phosphorylation processes. We focus on recent dynamics which involve allosteric mechanisms, and show, e.g., how one can relate the multi-stability of an underlying ordinary differential equation to ultrasensitivity of the stochastic system. The examples cover allosteric phosphorylation processes, substrate-catalyst chain reactions and nucleosome mediated cooperativity.

We present basic models for processes involving the binding of ligand molecules on macromolecules, or for phosphorylation processes where molecules can be phosphorylated at multiple sites. Let $M$ be a macromolecule containing $N$ sites $S=\{1,\cdots,N\}$ where ligand molecules can bind. We will use the binary variables $n_i = 0,\ 1$, $i=1,\cdots,N$ to describe the occupancy of the various sites: $n_i =1$ means that site $i$ is occupied or phosphorylated, while $n_i =0$ indicates that no molecule is bound at site $i$. The configuration space is denoted by $\Lambda = \{n=(n_i)_{1\le i\le N};\ n_i =0,\ 1\}$, which has size $\vert\Lambda\vert = 2^N$. We suppose that the ligand concentration is given by a positive variable $v>0$, and that the probability $\pi(n)$ to see a configuration $n$ is of the generic form
\begin{equation}\label{BasicProbability}
\pi(n)=\frac{\mu(n)v^{\vert n\vert}}{Z(v)},
\end{equation}
where the $\mu(n)$ are non-negative weights, $\vert n\vert$ denotes the number of bound sites, that is,
$$\vert n\vert =\sum_{i=1}^N n_i$$
and $Z(v)$ is the normalization constant $Z(v)=\sum_{n\in\Lambda}\mu(n)v^{\vert n\vert}$.  Let $a$ be a non-negative and increasing function defined on the unit interval $[0,1]$. The activity of the macromolecule is defined as the statistical average
\begin{equation}\label{Activity}
f(v)=\langle a(\tfrac{\vert n\vert}{N})\rangle_\pi = \sum_{n\in\Lambda}a(\tfrac{\vert n\vert}{N})\pi(n),
\end{equation}
see e.g.~\cite{ryerson2014ultrasensitivity}, which is a non-decreasing function of $v>0$.

As a matter of clarity, notice that the statistical average in (\ref{Activity}) is taken over all possible values of $n$, but can be defined with respect to $\bar\pi$, the law of $\vert n\vert$, which is a probability measure defined on the set $\bar\Lambda =\{0,1,\cdots,N\}$, with
\begin{equation}\label{Image0}
\bar\pi(k)=\sum_{n:\ \vert n\vert =k}\pi(n)= v^k \frac{\sum_{n:\ \vert n\vert =k}\mu(n)}{Z(v)}\,.
\end{equation}
In this case, the activity of the macromolecule becomes $f(v)=\sum_{k=0}^N a(\tfrac{\vert n\vert}{N})\bar\pi(k)$. When the weights $\mu(n)$ only depend on $n$ through the number of bound sites $\vert n\vert$, that is, when $\mu(n)=V(\vert n\vert)$ for some function $V$, then
\begin{equation}\label{Image}
\bar\pi (k)=\frac{{N\choose k}V(k)v^k}{\sum_{k=0}^N {N\choose k}V(k)v^k}\,.
\end{equation}

The {\bf Hill coefficient} of cooperativity 
\begin{equation}\label{HillCoefficient}
\eta_H(v)=v \frac{\,{\rm d}}{\,{\rm d}v}\ln\Big(\frac{f(v)}{f(\infty)-f(v)}\Big),
\end{equation}
provides a measure of the effect of binding of one ligand molecule at some site on the binding at other sites, see e.g.~\cite{MazzaBenaim} and the references therein for more details. Let $\bar v$ be a concentration such that $f(\bar v)$ is halfway between the minimum and the maximum of $f$. One speaks of {\bf positive cooperativity} when $\eta_H(\bar v)$ is larger than one, and of { \bf ultrasensitivity} when $\eta_H(\bar v)$ is very large. 
We will study the large $N$ behaviour of $\eta_H(v)$ as a function of the ligand concentration $v$, and look for critical concentrations $v_c$ ensuring that $\eta_H(v)$ diverges towards infinity as $N\to\infty$. This defines our generic notion of ultrasensitivity.
But of course, this particular choice of concentration $\bar v$ is more pragmatic than based on scientific grounds. In the current work, we will look for concentrations $v_c$ for which $\eta_H(v_c)$ is high and study the dependence between this coefficient and the number of binding sites. The biochemical literature also considers a second measure of cooperativity, which is sometimes called the Goldbeter-Koshland index or the {\bf effective Hill coefficient}, see e.g.~\cite{Karlin1980a,Karlin1980b}. In the next section, we provide an extended definition of this index and establish its link with the Hill coefficient $\eta_H(v_c)$.

\section{New formulas for Hill coefficients}
\subsection{Statistical interpretation of Hill coefficients}
The Hill coefficient  (\ref{HillCoefficient}) has a nice statistical representation. For $\pi$ be a probability measure of the form given in (\ref{BasicProbability}), we prove in the Appendix the following result.
\begin{theorem}\label{SecondFormula}
Let $\pi$ be a probability measure of the form given in~\eqref{BasicProbability}. Then
\begin{equation}\label{Formula1}
\eta_H(v)=\frac{{\rm Cov}_\pi(a(\tfrac{\vert n\vert}{N}),\tfrac{\vert n\vert}{N}) a(\frac{k^*}{N})}{f(v)(a(\frac{k^*}{N})-f(v))} \, N.
\end{equation}
where $k^*$ denotes the largest $k$ for which there is a configuration $n$ such that $\vert n\vert =k$ and $\mu(n)>0$, in such a way that $a(\frac{k^*}{N})=f(\infty)$. Moreover, 
$\lim_{v\to\infty}\eta_H(v) = \lim_{v\to0}\eta_H(v) =1$.
\end{theorem}
In the special case, $a(x)\equiv x$ and $k^* =N$, one gets the formula
\begin{equation}\label{Formula2}
\eta_H(v)= \frac{{\rm Var}_\pi(\frac{\vert n\vert}{N})}{\bar p (1-\bar p)} \,N,
\end{equation}
where $\bar p = \sum_i \pi(n_i=1)/N$.

\subsection{Effective Hill coefficients}
The ratio $0\le f(v)/f(\infty)\le 1$ can be seen as a probability distribution function. Let $q\in [1/2,1]$, and consider the quantiles $v_q$ and $v_{1-q}$ defined as $q=f(v_q)/f(\infty)$ and $1-q=f(v_{1-q})/f(\infty)$. The Goldbeter-Koshland index, or the effective Hill coefficient, is defined by
\begin{equation}\label{Koshland}
I_q = \frac{2\ln(\frac{q}{1-q})}{\ln(\frac{v_q}{v_{1-q}})}.
\end{equation}
The standard definition corresponds to the special choice $q=0.9$, and the related index $I_{0.9}$ provides a measure of the dose difference one must consider to move $f(v)/f(\infty)$ from a low 10\% saturation level to a high 90\% saturation level. Steep curves have high $I_q$, or $v_q/v_{1-q}$ close to 1. The Hill coefficient and the effective Hill coefficient are related as follow. Let $\bar \eta_H$ be such that $\eta_H(v)=\bar\eta_H(\ln(v))$. Then
\begin{equation}\label{Relation}
I_q = \frac{1}{\ln(v_q)-\ln(v_{1-q})}\int_{\ln(v_{1-q})}^{\ln(v_q)}\bar\eta_H(y)\,{\rm d}y,
\end{equation}
see e.g.~\cite{MazzaBenaim} for more details.

\section{Influence functions and sharp-thresholds}

As seen above, Hill coefficients $\eta_H(v)$ and their effective versions $I_q$ are used to measure the steepness of binding curves in biological problems. Efficient genetic switches occur when the binding curve switches abruptly from a low saturation level to a high saturation level within a small concentration interval at the log scale. Similar switches occur in many frameworks of applied probability and statistical mechanics, like reliability theory, random graph theory and percolation theory, where sharp-threshold phenomena are common place. A well developed theory to study such phenomena exists, see, e.g., \cite{ben1990,bourgain1992,grimmett2003,grimmett2006,friedgut2004}; we will make explicit links between these fields in what follows. These results give general conditions ensuring the emergence of ultrasensitivity in systems biology. 

\subsection{Site-specific Hill coefficients and conditional influences}

Site-specific Hill coefficients $\eta_{H,i}(v)$ are defined to measure the effect of the binding of a molecule at site $i$ on the binding of molecules at sites different from $i$, see \cite{dicera} and \cite{MazzaBenaim}. More formally,
\begin{equation}\label{specificHill}
\eta_{H,i}(v)=1+\E_\pi(\bar n_i\vert n_i=1)-\E_\pi(\bar n_i\vert n_i=0),
\end{equation}
where $\E_\pi(\cdot\vert n_i=\varepsilon)$ is the conditional expectation
under the probability measure $\pi$ conditional to the event $\{n_i=\varepsilon\}$, $\varepsilon =0,1$, and where $\bar n_i =\sum_{j\ne i}n_j$  is the ligation number at sites different from $i$. $\eta_{H,i}(v)$ gives thus the gain in bound molecules at site different from $i$ when adding a molecule at site $i$, and is larger than 1 for cooperative biochemical systems. 

Assume that $a(x)\equiv x$ and that $\mu(n)>0$, $\forall n$. Then, see \cite{MazzaBenaim},
$$\eta_H(v)=\frac{{\rm Var}_\pi(\vert n\vert)}{N\bar p (1-\bar p)}
=\frac{\sum_i p_i(1-p_i)\eta_{H,i}(v)}{N\bar p (1-\bar p)},$$
where $p_i = \E_\pi(n_i)=\pi(n_i=1)$ and $\bar p=\sum_i p_i/N$. We follow \cite{grimmett2006}; let $\mu$ be a positive probability measure on $\Lambda$, and define, for $0<p<1$, the new probability measure
\begin{equation}\label{SecondProbability}
\mu_p(n)=\frac{1}{Z_p}\mu(n)\prod_i \Big(p^{n_i}(1-p)^{1-n_i}\Big),
\end{equation}
where $Z_p$ is the normalization constant or partition function. Then $\mu_p$ coincides with $\pi$ when the concentration $v$ and $p$ are such that
 $v=p/(1-p)$. Let $A\subset \Lambda$ be a subset of $\Lambda$. The 
 {\bf conditional influence} $I_A(i)$ is defined as
 \begin{equation}\label{ConditionalInfluence}
 I_A(i)=\mu_p(A\vert n_i=1)-\mu_p(A\vert n_i=0),
 \end{equation}
 that is,
 $$I_A(i)=\E_{\mu_p}({\rm I}_A\vert n_i=1)-\E_{\mu_p}({\rm I}_A \vert n_i=0),$$
 where ${\rm I}_A$ is the indicator function of the subset $A$. One sees that
 $$\eta_{H,i}(v)=1+\sum_{j\ne i}I_{\{j\}}(i),$$
 when $v=p/(1-p)$. 
 An event $A\subset \Lambda$ is called increasing if $n\in A$
 whenever there exists $n'\in A$ such that
 $n\ge n'$. When $\mu_p$ is a product measure with
 $p_i\equiv p$, and $A$ is an increasing event, there exists an absolute positive constant $c$ such that, $\forall N$, $p\in (0,1)$, there exists
 $i\in [N] = \left\{1,\dots,N \right\}$ such that
 \begin{equation}\label{Inequality}
 I_A(i)\ge c \min\{\mu_p(A),1-\mu_p(A)\}\frac{\ln(N)}{N},
 \end{equation}
 see \cite{kahn1988,bourgain1992,friedgut1996,talagrand1994}.
 \subsection{Conditional influence functions and sharp-thresholds}
 The aim of the sharp-threshold theory is to give conditions ensuring that 
 the function $\mu_p(A)$ exhibits a sharp-threshold as $p$ varies within a small interval of values of $p$ of size $1/\ln(N)$. Such conditions are obtained using a Russo-type formula (see \cite{grimmett2003,grimmett2006}) of the form
 \begin{equation}\label{Russo}
 \frac{{\rm d}\mu_p(A)}{{\rm d}p}=\frac{1}{p(1-p)}{\rm Cov}_{\mu_p}({\rm I}_A,\vert n\vert),
 \end{equation}
 which is similar to our equation (\ref{Formula1}). A direct computation shows that (\ref{Russo}) is a special instance of (\ref{Formula1})
 for the special activity function $a(x)={\rm I}_A(x)$. 
 We follow next \cite{grimmett2006} to introduce various probabilistic notions and 
 a powerful Theorem that yields results on sharp-thresholds.
 For $J\subset S$ and $\xi\in\Lambda$, let $\Lambda_J =\{0,1\}^J$ and
 $$\Lambda_J^\xi = \{n\in\Lambda;\ n_j =\xi_j\hbox{ for }j\in S\setminus J\}.$$
 The set of all subsets of $\Lambda_J$ is denoted by ${\cal F}_J$. 
 Let $\mu$ be a positive probability measure on $(\Lambda,{\cal F}_S)$. The conditional probability measure $\mu_J^\xi$ on $(\Lambda_J,{\cal F}_J)$
 is defined by
 $$\mu_J^\xi(n_J)=\mu(n_J\vert n_i =\xi_i \hbox{ for }i\in S\setminus J),\ n_J\in \Lambda_J.$$
The probability measure $\mu$ is said to be {\it monotonic} when, for all
$J\subset S$, all increasing subsets $A\subset \Lambda_J$, and all $\xi\in\Lambda$,
\begin{equation}\label{monotone}
\mu_J^\xi(A)\le \mu_J^\eta(A) \hbox{ whenever }\xi \le \eta.
\end{equation}
It turns out that $\mu$ is monotonic if and only if it is {\it 1-monotonic}, that is, if (\ref{monotone}) holds for all singleton sets $J$.  The following result from 
\cite{grimmett2006} is very useful to obtain results on sharp-thresholds:
\begin{theorem}\label{ThSharp}
There exists a positive constant $c$ such that the following holds. Let $A\in {\cal F}_S$ be an increasing event. Assume that $\mu_p$ is monotonic for all $p$. If there exists a subgroup of the permutation group of $N$ elements that acts transitively on $S$ leaving both $\mu$ and $A$ invariants, then
\begin{equation}\label{inequality}
 \frac{{\rm d}\mu_p(A)}{{\rm d}p}\ge \frac{c \alpha_p}{p(1-p)}
 \min\{\mu_p(A),1-\mu_p(A)\}\ln(N),
 \end{equation}
 where
 $\alpha_p = \mu_p(n_i)(1-\mu_p(n_i))$.
 \end{theorem}

  This implies that, for $0<\varepsilon <1/2$, the function $f(p)=\mu_p(A)$ increases from 
 $\varepsilon$ to $1-\varepsilon$ over an interval of values of $p$
 with length smaller in order than $1/\ln(N)$, which is precisely a sharp-threshold, which implies that the quantiles $v_q$ and $v_{1-q}$ with
 $q=1-\varepsilon$ are such that $v_{q}-v_{1-q}\le 1/\ln(N)$ leading to an ultrasensitive behaviour.
 
 \begin{example}\label{BoltzmannMachine}
 A basic model which describes interactions between binding sites is the Boltzmann machine model or the {\bf Ising model}. Consider the free energy function
 \begin{equation}\label{BoltzmannMachine}
H(n) = -\sum_{i\ne j}J_{ij}n_i n_j -\sum_i h_i n_i,
\end{equation}
where the coefficients $J_{ij}=J_{ji}$ model pairwise interactions, and where the parameters $h_i$
are local field. One then defines the related Gibbs distribution
$$\pi_\beta(n)=\frac{1}{Z(\beta)}\exp(-\beta H(n)),$$
where $\beta > 0$ is the inverse temperature. Such models
appear in systems biology when modelling transcription rates,
see, e.g., \cite{hwa2003,MazzaBenaim} and the references therein.
Usually, it is defined using a graph of node set 
$S=\{1,\cdots,N\}$, and of edge set
${\cal E}=\{e=(i,j);\ J_{ij}\ne 0\}$. 
The model is said to be ferro-magnetic  when
$J_{ij}\ge 0$. In this situation,  assuming that  $h_i\equiv 0$ and $v>1$, the most probable
configuration is the fully occupied one with $n_i\equiv 1$.
The Gibbs distribution corresponds to $\mu_p$ when $\exp(\beta h_i)\equiv v=p/(1-p)$ and $\mu(n)=\exp(\beta \sum_{i\ne j}J_{ij}n_i n_j)$. Any event of the form $A=\{n\in\Lambda;\ \vert n\vert > \theta\}$ is increasing.
These events are often used to define promoter activities, and thus can be used to model transcriptional activity.
 The previous results show the existence of a sharp-threshold phenomenon for the probability
$\pi_\beta (\vert n\vert > \theta)$.
\end{example}

\section{Density dependent birth and death processes}

We will concentrate on probability distribution $\bar\pi_N$ on $\{0,1,\cdots,N\}$ which are steady state distributions of density dependent birth and death processes. The classical biochemical literature uses a slightly different language. We present first the basic framework in this setting following \cite{enciso2013}, and use next a more probabilistic approach. Let $c_k(t)$ denote the concentration of molecules with exactly $k$ modified sites for phosphorylation processes, or, with exactly $k$ bound ligand molecules for binding processes. The time evolution of these concentrations is often assumed to be of the form
\begin{equation}\label{Standard}
\frac{{\rm d}c_k}{{\rm d}t}= b_{k-1} c_{k-1} -d_k c_k -b_{k} c_k +d_{k+1}c_{k+1},
\end{equation}
where $b_k$ depends linearly on the inducer concentration $v$. In the above equation it is assumed that $c_k$ turns into $c_{k+1}$ at the linear rate $b_k c_k$ and that $c_{k+1}$ turns back into $c_k$ at rate $d_{k+1}c_{k+1}$. This equation can be seen as the Kolmogorov forward equation associated with a birth and death process $Y_N(t)$ of birth rate $q_N(k,k+1)= b_k$ and death rate $q_N(k,k-1)=d_k$ (see e.g.~\cite{MazzaBenaim}),. This means that, for $h$ small,
$$\Prob\left(Y_N(t+h)=k+1\mid Y_N(t)=k\right)\sim b_k h,$$
$$\Prob\left( Y_N(t+h)=k-1\mid Y_N(t)=k\right)\sim d_k h.$$
Such processes are said to be {\bf density dependent} when furthermore
$$q_N(k,k+1)=N b^{(N)}(\tfrac{k}{N})\quad\hbox{ and }\quad q_N(k,k-1)=N d^{(N)}(\tfrac{k}{N}),$$
for some functions $b^{(N)}$ and $d^{(N)}$. Here we concentrate on functions $b^{(N)}$ and $d^{(N)}$, the birth and death rate respectively, given by two Lipschitz-continuous functions on $[0,1]$, such that $b^{(N)}>0$ and $d^{(N)}>0$ on $]0,1[\,$, and $b^{(N)}(1)=d^{(N)}(0)=0$. The steady state distribution $\bar\pi_N$ is then given by
\begin{equation}\label{Steady}
\bar\pi_N(\frac{k}{N}) = \bar\pi_N(0)\prod_{j=0}^{k-1}\frac{b^{(N)}(\frac{j}{N})}{d^{(N)}(\frac{j+1}{N})}\,, \qquad k=1,\ldots,N,
\end{equation}
with
\begin{equation}\label{Steady:at:0}
\bar\pi_N(0)^{-1}=\sum_{k=0}^N \prod_{j=0}^{k-1}\frac{b^{(N)}(\frac{j}{N})}{d^{(N)}(\frac{j+1}{N})}\,.
\end{equation}
We assume furthermore that $b^{(N)}$ depends linearly on the inducer concentration $v$, so that the steady state distribution $\bar\pi_N$ has the form defined in~\eqref{Image0}. 
Assume for simplicity that
$$b^{(N)}-d^{(N)}\xrightarrow[N\to\infty]{} F \quad\hbox{ and } \quad\ln\left(\frac{d^{(N)}}{b^{(N)}}\right)\xrightarrow[N\to\infty]{} \ln(r),$$ for some well-behaved functions $F$ and $r$. One can check that the renormalized process $X_N(t)=Y_N(t)/N$ converges as $N\to\infty$ towards the orbits of the ordinary differential equation (o.d.e.)
\begin{equation}\label{ODE}
\frac{\,{\rm d}x(t)}{\,{\rm d}t}=F(x(t)),\ \ x(0)=x_0.
\end{equation}
when $X_N(0)\longrightarrow x_0$, as $N\to\infty$. The free energy function $J$ and the entropy function $I$ are defined by 
\begin{equation}\label{FreeEnergy}
J(x)=\int_0^x \ln(r(u)){\rm d}u \quad\hbox{ and }\quad I(x)=J(x)-J_0,
\end{equation}
where $J_0 =\min_{x\in [0,1]}J(x)$.
\cite{chan1998largedeviation} proved that the family of steady state distributions
$\bar\pi_N$ satisfies a large deviation principle of rate function $I$.
We will obtain precise large deviations by showing that, under some assumptions, one can find positive constants $0<\gamma_- <\gamma_+$ such that
\begin{equation}\label{Precise}
\frac{\gamma_- }{\sqrt{N}}\exp(-N I(\frac{j}{N}))
\le \bar\pi_N(\frac{j}{N})\le 
\frac{\gamma_+ }{\sqrt{N}}\exp(-N I(\frac{j}{N})),
\end{equation}
showing that the steady state distribution of the process concentrates
asymptotically in the neighbourhood of the global minima of the entropy function $I$ (see Lemma~\ref{lemma:pi(j/N)} in the Appendix).
 When $a(x)\equiv x$, the Hill coefficient associated with the steady state distribution $\bar\pi_N$ of the process is given by
 \begin{equation}\label{Hill}
\eta_H(v)= \frac{{\rm Var}_{\bar\pi_N}(X_N)}{\langle X_N \rangle_{\bar\pi_N}(1-\langle X_N \rangle_{\bar\pi_N})}\, N,
\end{equation}
The systems will be hence highly ultrasensitive with a Hill coefficient of order $N$ when the steady state variance ${\rm Var}(X_N)$ converges towards a constant. Similarly, for general activity functions $a(x)$, the Hill coefficient is given by the covariance
\begin{equation}\label{Formula1:BD}
\eta_H(v)=\frac{{\rm Cov}_{\bar\pi_N}(X_N,a(X_N)) a(\frac{k^*}{N})}{f(v)(a(\frac{k^*}{N})-f(v))}\, N.
\end{equation}

\subsection{Sharp threshold}
 When the sites are identical, $\bar\pi_N(\tfrac{\vert n\vert}{N}) = \bar \pi(\vert n \vert) = {N \choose {\vert n \vert}}\pi(n)$ (see (\ref{Image0}) and (\ref{Image})) and the related probability measure is monotonic when
\begin{equation}\label{eqn:bd:monotonic}
\frac{k+2}{N-k-1}d^{(N)}\lp k+1 \rp b^{(N)}\lp k+1 \rp \ > \ \frac{k+1}{N-k} d^{(N)}\lp k+2 \rp b^{(N)} \lp k \rp,
\end{equation}
for all $k=0,..,N-2$. 
Section \ref{s.examples} provides three biological examples where the underlying steady state distribution $\pi$ is monotonic. In these examples, the limiting 
(o.d.e.) (\ref{ODE}) possesses two stable equilibria $0<x_1<x_2<1$. Moreover, we will show that there exists a critical concentration $v_c$ such that
the entropy function $I$ attains its global minimum at $x_1$ when $v<v_c$, at
$x_2$ when $v>v_c$, and at $x_1$ and $x_2$ when $v=v_c$. Consider the increasing
event $A=\{n;\ \vert n\vert > \kappa N\}$, for some positive threshold $\kappa > 0$.
The conditions of Theorem \ref{ThSharp} are satisfied, and one gets that the derivative is larger than a constant times $\ln(N)$ when the pre-factor
$\min\{\pi(A),1-\pi(A)\}$ which appears in the inequality (\ref{inequality}) is asymptotically positive. As stated previously, the steady state satisfies a large deviation principle of rate function $I$, so that the pre-factor vanishes asymptotically exponentially fast when either $\kappa < x_1$ or $\kappa > x_2$.
In the intermediate situation where $x_1<\kappa<x_2$, the behaviour of the factor depends on the concentration $v$. If $v<v_c$, $I$ attains its minimum at $x_1$, so that $\pi(A)$ converges exponentially fast toward 0. When $v>v_c$, $I$ attains its minimum at $x_2$, and $\pi(A)$ converges exponentially fast towards 1. When however $v=v_c$, both probabilities $\pi(A)$ and $1-\pi(A)$ are asymptotically positive, and (\ref{inequality}) yields that the system exhibits a sharp-threshold phenomenon in the neighbourhood of $v=v_c$.

\subsection{Multi-stability and ultrasensitivity}

 When the (o.d.e) (\ref{ODE}) possesses a single stable equilibrium $x_1\in (0,1)$, the law of the stationary process $X_N$ is concentrated around $x_1$ when $N$ is large, and the variance converges toward 0. This is the monostable case, which does not lead to ultrasensitivity with a Hill coefficient of order $N$ since, from (\ref{Hill}), $\eta_H(v)/N\longrightarrow 0$. When the system is multi-stable, that is, when the (o.d.e.) has at least two different stable equilibria, the steady state distribution $\bar\pi_N$ concentrates on any neighborhood of a stable equilibrium $x_i$ when the entropy function $I$ attains a global minimum at $x_i$, with $I(x_i)=0$. When this occurs for at least two different stable equilibria $x_1 < x_2$, the variance ${\rm Var}_{\bar\pi_N}(X_N)$ is positive in the large $N$ limit, so that, from (\ref{Hill}), 
$$\liminf_{N\to\infty}\frac{\eta_H(v_c)}{N}\neq0,$$
 see Fig.~\ref{fig:1} and Fig.~\ref{fig:hillEnciso}. The above picture is a consequence of Theorem \ref{AsymptoticHill}, which is obtained under some assumptions on the birth and death process.
\begin{assump}\label{assump:r}
\begin{enumerate}
\item The (o.d.e.) (\ref{ODE}) possesses a finite set of $m>1$ equilibria $\lac x_{1},\ldots,x_{m}\rac$, such that $I(x_i)=0$, $I'(x_i)=0$ and $I''>0$ in a neighborhood of each $x_i$ for all $i$.
\item There exists a function $r$ on $[0,1]$ such that
\begin{equation}\label{Condition1}
\int_0^1 \vert \ln(r(u)) \vert \,{\rm d}u < \infty.
\end{equation}
\item The limiting rates $b(x)$ and $d(x)$ satisfy
\begin{eqnarray*}\label{eq:condition:bx:dx}
\lim_{x\to 0+}\frac{b(x)}{1-x} = \ell_{b,0}, & \quad & \lim_{x\to 0+}\frac{d(x)}{x} = \ell_{d,0},\\
\lim_{x\to 1-}\frac{b(x)}{1-x} = \ell_{b,1}, & \quad & \lim_{x\to 1-}\frac{d(x)}{x} = \ell_{d,1}	
\end{eqnarray*}
where $\ell_{b,0}, \ell_{b,1}, \ell_{d,0}, \ell_{d,1}$ are four positive constants.
\item
The function $r$ is left-continuous in $]0,1[$ and piecewise ${\cal C}^2$. Namely, there exists a finite number of discontinuities $\lac d_{1}, \ldots,d_{K}\rac$ of $r, r' $ and $r''$ in $]0,1[$.
\item There exists a finite number $M$ such that
\begin{equation}\label{NewCondition}
M=\limsup_{N\to\infty} \max_{1\leq j\leq N}\left\vert \sum_{k=1}^j \ln\left(\frac{d^{(N)}(\frac{k}{N})}{b^{(N)}(\frac{k}{N})}\right)-\sum_{k=1}^j \ln(r(\tfrac{k}{N}))\right\vert.
\end{equation}
\end{enumerate}
\end{assump}

\begin{theorem}\label{AsymptoticHill}
Suppose that Assumptions~\ref{assump:r} are satisfied and assume that the activity function $a$ is continuous, bounded and strictly increasing on $[0,1]$. Assume furthermore that $\bar\pi_N$ is of the form given in~\eqref{BasicProbability}. 
\begin{enumerate}
\item Case $m=1$. If the entropy function $I$ attains a minimum at a unique point $0<x_1<\lim \frac{k^*}{N}$, then
$$\frac{\eta_H(v)}{N}\xrightarrow[N\to\infty]{} 0.$$

\item Case $m>1$. If there is a critical concentration $v_c$ such that $I$ attains its minimum at, at least, two different points $x_1\ne x_2$ of the unit interval, then
$$\liminf_{N\to\infty}\frac{\eta_H(v_c)}{N} > 0,$$
and the system exhibits ultrasensitivity of order $N$.
\end{enumerate}
\end{theorem}

The general idea of the proof, whose details are postpone to the Appendix, is to compute precise large deviations of the steady state measure $\bar\pi_N$ in order to establish that it charges all global minima of the entropy function $I(x)$ (see Lemma~\ref{lemma:pi:charges:equilibr}), so that it converges weakly to a combination of Dirac measures (see Corollary~\ref{cor:WeakConvergence}).

We illustrate in the next sections using well chosen examples that ultrasensitivity usually occurs for a critical concentration $v_c$: when $v\ne v_c$, the multi-stable system has only one equilibrium minimizing $I$, so that the steady state distribution is asymptotically unimodal with a low Hill coefficient, while when $v=v_c$, the steady state distribution has asymptotically two modes which correspond to the two equilibria of the (o.d.e.), with a large Hill coefficient.

\begin{figure}
\centering
\includegraphics{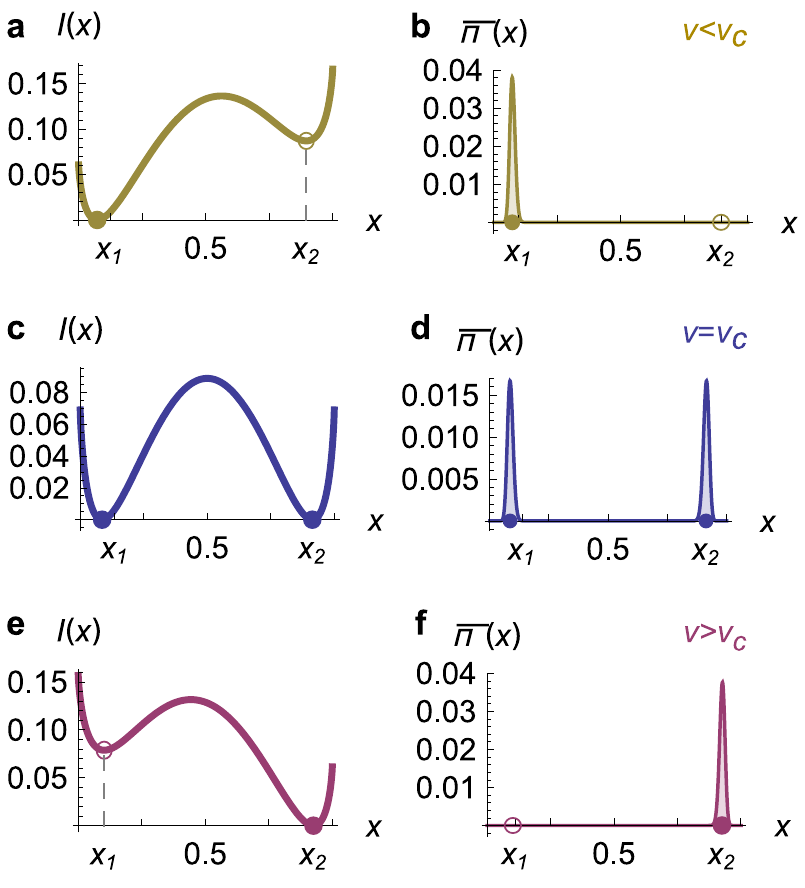}
\caption{Bistable systems lead to bimodal steady state distributions and ultrasensitivity when $I$ possesses two global minima. (\textbf{a}) The stable equilibria $x_1$ and $x_2$ of a bistable system are the local minima of the entropy function $I$. (\textbf{b}) When $v < v_c$, the unique equilibrium minimizing $I$ is the mode of the related unimodal steady state distribution. (\textbf{c}) When $v=v_c$, both equilibria minimize $I$. (\textbf{d}) The steady state distribution is bimodal of modes $x_1$ and $x_2$. In this last case, the systems exhibits ultrasensitivity, with a Hill coefficient $\eta_H(v_c)$ linear in the number $N$ of binding sites. (\textbf{e}) When $v>v_c$ the unique equilibrium minimizing $I$ is $x_2>x_1$. (\textbf{f}) The steady state distribution is unimodal. }\label{fig:1}
\end{figure}

\begin{figure}[h!]
\centering
\includegraphics{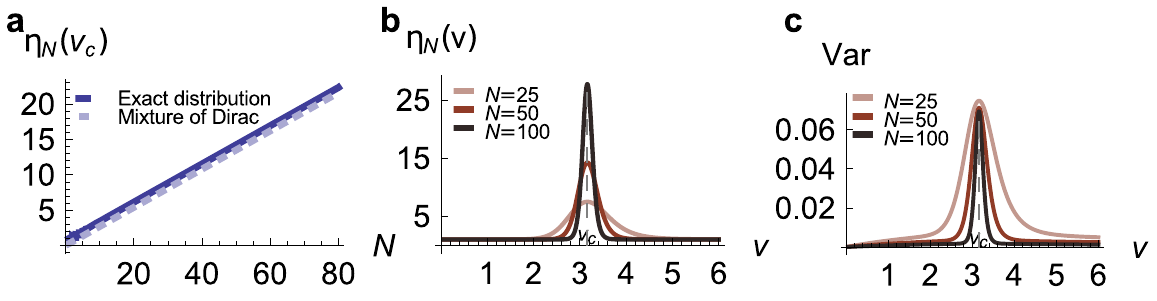}
\caption{Hill coefficient is maximal for $v=v_c$ and the system is ultrasensitiv. (\textbf{a}) Hill coefficient versus $N$ for $\varepsilon = 0.1$ and $a(x) = \frac{x}{1+x}$. The continuous blue curve denotes the true value of the Hill coefficient. The dotted line plots the Hill coefficient obtained by simulating a mixture of Dirac random variables $\mu = \frac{1}{2} \delta_{x_1} + \frac{1}{2} \delta_{x_2}$. (\textbf{b}) The Hill coefficient $\eta_h(v)$ is represented for different values of $v$. It attains a maximum for $v=v_c$. (\textbf{c}) The variance $\textrm{Var}_{\bar\pi_N}(X_N)$ is represented for different values of $N$. It is maximal for $v=v_c$ and concentrates on this value as $N$ grows.}\label{fig:hillEnciso}
\end{figure}

\section{Illustrations from systems biology\label{s.examples}}

In the following sections we illustrate how our results apply on different models from systems biology. Allosteric phosphorylation processes~\cite{EncisoKelloggVargas2014}, substrate-catalyst interactions~\cite{Hatakeyama:2014aa} as well as nucleosome-mediated cooperativity~\cite{mirny2010} are investigated.

\subsection{Allosteric phosphorylation processes\label{s.allosteric}}
Let us consider the model proposed in \cite{EncisoKelloggVargas2014} where the protein is either active (A) or inactive (I), and has $N$ sites that can be phosphorylated. The transition rates are given in Fig.~\ref{fig:strip}, see also \cite{tenWolde} where the reason for taking $\varepsilon^k$ in the switching rates is motivated using free energies.
This is an adaptation of the classical Monod-Wyman-Changeux (MWC) model \cite{monod1965} which is one of the first model where ultrasensitivity was considered. We follow \cite{enciso2013,ryerson2014ultrasensitivity,EncisoKelloggVargas2014} using a probabilistic framework. 
 
Let $W(t)$ be the Markov chain associated with the protein activity ($W(t)\in\lac I, A\rac$). The number of phosphorylated sites at time $t$ is described by a process $N(t)$, so that the full process is given by a bivariate time-continuous Markov chain $(N(t),W(t))$. The authors of \cite{enciso2013,ryerson2014ultrasensitivity,EncisoKelloggVargas2014} opt for Markov chains of transition rates
$$q_N((k,A),(k+1,A))=(N-k)v\quad \text{ and }\quad q_N((k,A),(k-1,A))=k,$$
in the active state, and similarly in the inactive state
$$q_N((k,I),(k+1,I))=\varepsilon(N-k)v\quad \text{ and }\quad q_N((k,I),(k-1,I))=k,$$
where the small parameter $\varepsilon < 1$ models the low affinity associated with the inactive state. The transition between the active and the inactive state are given by 
$$q_N((k,I),(k,A))=L_2\quad \text{ and }\quad q_N((k,A),(k,I))=L_1\varepsilon^k.$$
In some cases, the steady state distribution of $N(t)$ is explicit, as illustrated in following Proposition, which is proven in the Appendix.

\begin{proposition}\label{prop:mixture0}
Assume $L_2 =1$. The marginal distribution of $N(t)$ is then given by
\begin{equation}\label{Mixture0}
\bar\pi(k) =\frac{{N\choose k}(\frac{v}{1+v})^k (\frac{1}{1+v})^{N-k}}{1+(\frac{1+\varepsilon v}{1+v})^N L_1}
      +\frac{{N\choose k}(\frac{\varepsilon v}{1+\varepsilon v})^k (\frac{1}{1+\varepsilon v})^{N-k}}{1+(\frac{1+ v}{1+\varepsilon v})^N L_1^{-1}},
\end{equation}
which is a mixture of the Binomial distributions $\mathcal{B}(N,\frac{v}{1+v})$ and $\mathcal{B}(N,\frac{\varepsilon v}{1+\varepsilon v})$. 
\end{proposition}

\subsubsection{Ultrasensitivity in the Hill sense}

We suppose here that $L_1 = L_1(N)=\varepsilon^{-N/2}$ and $L_2\equiv 1$, as in \cite{ryerson2014ultrasensitivity,EncisoKelloggVargas2014}. Assume that the function $a$ is continuous, bounded and strictly increasing on the unit interval. We prove in the Appendix that when $v\neq v_c = 1/\sqrt{\varepsilon}$ and as $N\to\infty$, the Hill coefficient is asymptotically constant.
\begin{proposition}\label{asymptoticallostericHill}
Assume that the activity function $a$ is continuous, bounded and strictly increasing on the unit interval. When $v\neq v_c = 1/\sqrt{\varepsilon}$,
\begin{equation}\label{FormulaHillNonCrit}
\eta_H(v) \xrightarrow[N\to\infty]{} 
\begin{cases}
\frac{a(1)v a'\lp\frac{v}{1+v} \rp}{(1+v)^2 a\lp\frac{v}{1+v}\rp \lp a(1) - a\lp\frac{v}{1+v}\rp \rp} & \text{ if } v> v_c, \\
\frac{a(1)\varepsilon v a'\lp\frac{\varepsilon v}{1+\varepsilon v} \rp}{(1+\varepsilon v)^2 a\lp\frac{\varepsilon v}{1+\varepsilon v}\rp \lp a(1) - a\lp\frac{\varepsilon v}{1+ \varepsilon v}\rp \rp} & \text{ otherwise.}
\end{cases}
\end{equation}
When $v=v_c=1/\sqrt{\varepsilon}$, the asymptotic behavior of the Hill coefficient is given by
\begin{equation}\label{FormulaHill}
\eta_H(v_c)\sim C_{v_c}\, N,
\end{equation}
so that the system is ultrasensitive and where $C_{v_c}$ is a constant depending on $v_c$.
\end{proposition}

\subsubsection{Birth and death process approximation}

Consider the processes obtained by assuming fast switching rates between the active and inactive states, see e.g.~\cite[p.46]{MazzaBenaim}. Given that $N(t)=k$, the fast process $W(t)$ evolves according to the quasi-equilibrium 
\begin{align*}
\Prob\lp W(\infty)=I \mid N(\infty)=k\rp&=\sigma_k(I)=\frac{\varepsilon^{-\frac{1}{2}(N-2k)}}{1+\varepsilon^{-\frac{1}{2}(N-2k)}}\\
\Prob\lp W(\infty)=A\mid N(\infty)=k\rp&=\sigma_k(A)=\frac{1}{1+\varepsilon^{-\frac{1}{2}(N-2k)}}.
\end{align*}
The pair process $(N(t),W(t))$ is then replaced by a birth and death process $Y_N(t)$ of birth and death rates
 
 \begin{align*}
 q_N(k,k+1)&= N b^{(N)}(\tfrac{k}{N})=(N-k)v\sigma_k(A)+(N-k)v\varepsilon\sigma_k(I)\\
     &= v(N-k)\frac{1+\varepsilon \varepsilon^{-\frac{1}{2}(N-2k)}}{1+ \varepsilon^{-\frac{1}{2}(N-2k)}},
 \end{align*}
and
$$q_N(k,k-1)= N d^{(N)}(\tfrac{k}{N})=k,$$
see Fig.~\ref{fig:strip} (c), which has a steady state distribution $\bar\pi_N$ given by the binomial mixture (\ref{Mixture0}).

\begin{figure}
\begin{center}
\includegraphics{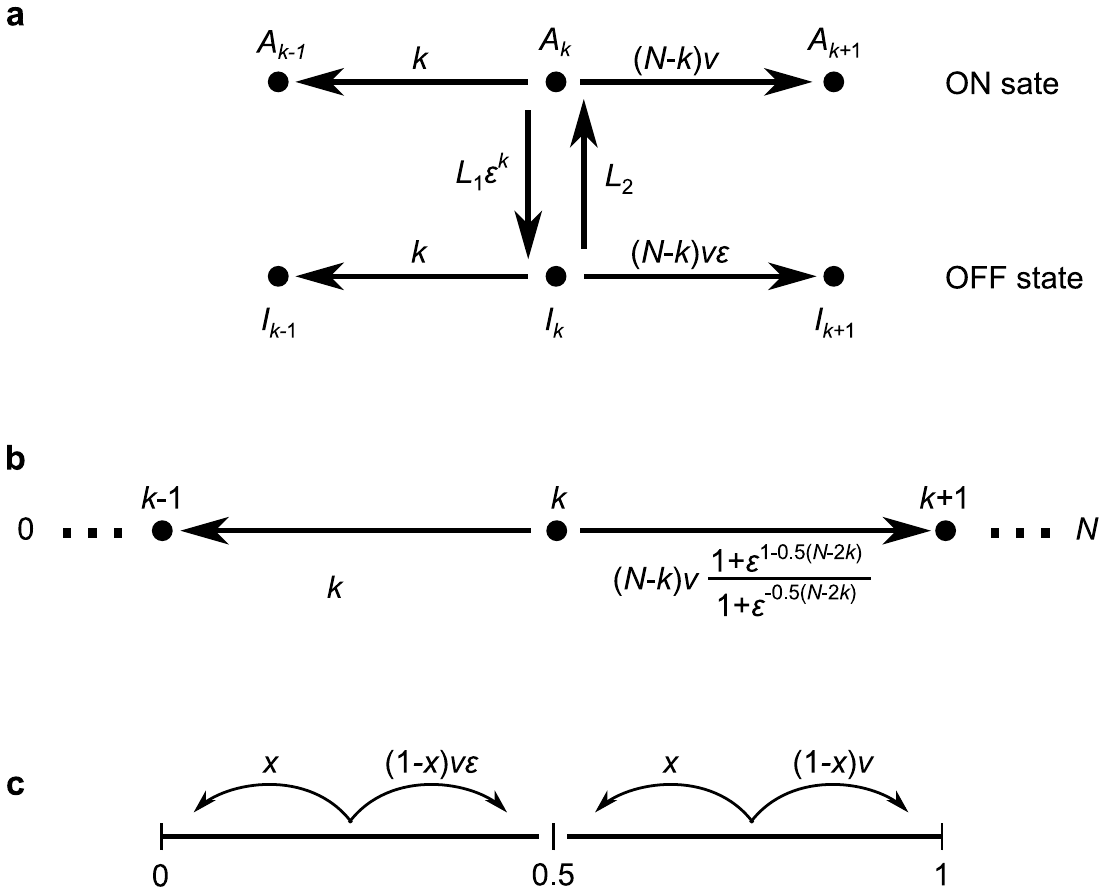}
\end{center}
\caption{Allosteric phosphorylation. (\textbf{a}) The process is a Markov Chain evolving along a strip. (\textbf{b}) Quasi-equilibrium approximation. (\textbf{c}) Density dependent birth and death process approximation when $N\to\infty$.}\label{fig:strip}
\end{figure}

The associated limiting (o.d.e.) (\ref{ODE}) is 
\begin{equation}\label{ODE1}
\frac{\,{\rm d}x}{\,{\rm d}t}= \begin{cases}
(1-x)\varepsilon v -x, &\mbox{ if } 0\le x<1/2,\\
(1-x) v -x, &\mbox{ if } 1/2\le x \le 1,\\
\end{cases}
\end{equation}
which possesses two stable equilibria $x_1=\varepsilon v/(1+\varepsilon v)$ and $x_2=v/(1+v)$, see Fig.~\ref{fig:strip}. 
 
\subsubsection{Entropy function and ultrasensitivity}

In this example, one can check that
$$\ln(r(x))=
\begin{cases}
 \ln\lp\frac{x}{1-x}\rp+\ln\lp\frac{1}{\varepsilon v}\rp, & \text{if } x <\frac{1}{2} \\
 \ln\lp\frac{1}{(1+\varepsilon) v}\rp, & \text{if } x=\frac{1}{2} \\
 \ln\lp\frac{x}{1-x}\rp+\ln\lp\frac{1}{v}\rp, & \text{if } x >\frac{1}{2}\,.
 \end{cases}
$$

\begin{figure}[h!]
\centering
\includegraphics{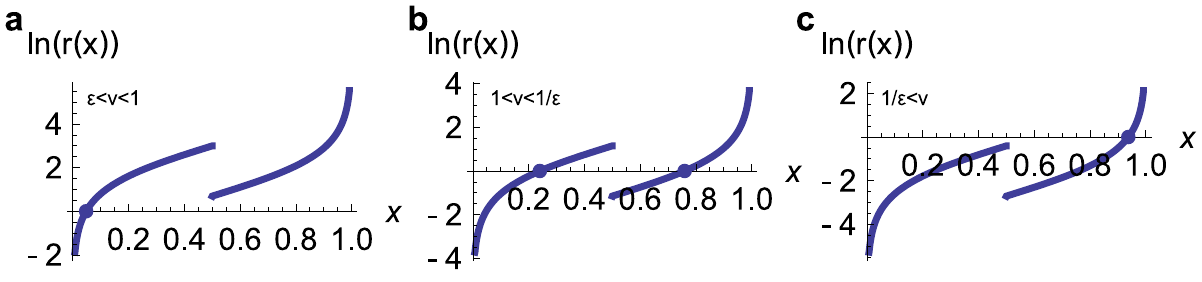}
\caption{Graph of $\ln(r(x))$ for different values of the parameters $(\varepsilon, v)$. (\textbf{a}) When $\varepsilon<v<1$ the function possesses one root in the interval $[0,1/2[$. (\textbf{b}) When $1<v<\frac{1}{\varepsilon}$ the function possesses one root in the interval $[0,1/2[$ and one root in the interval $]1/2,1]$. (\textbf{c}) When $\frac{1}{\varepsilon}<v$ the function possesses one root in the interval $]1/2,1]$.}\label{fig:logRx}
\end{figure}

The function $x\mapsto \ln(r(x))$ may have one or two zeros (see Fig.~\ref{fig:logRx}). The action functional
\begin{align}
J(x)&=\int_{0}^{x} \ln(r(y))\,{\rm d}y\nonumber\\
&=x\ln(x)+(1-x)\ln(1-x)-x\ln(v)-\ln(\varepsilon)\lp x\ind_{\lac x\leq\frac{1}{2}\rac}+\tfrac{1}{2}\,\ind_{\lac x>\frac{1}{2}\rac}\rp\label{eq:freeEnergy}
\end{align}
may thus possess one or two critical points in the interval $[0,1]$, which are located at $x_1=\frac{\varepsilon v}{1+\varepsilon v}$ and (or) $x_2=\frac{v}{1+v}$. The shape of the entropy function $I(x)=J(x)-J_{0}$ strongly depends on the value of the parameters (see Fig.~\ref{fig:EntropyFunction}).

\begin{figure} \centering
\includegraphics{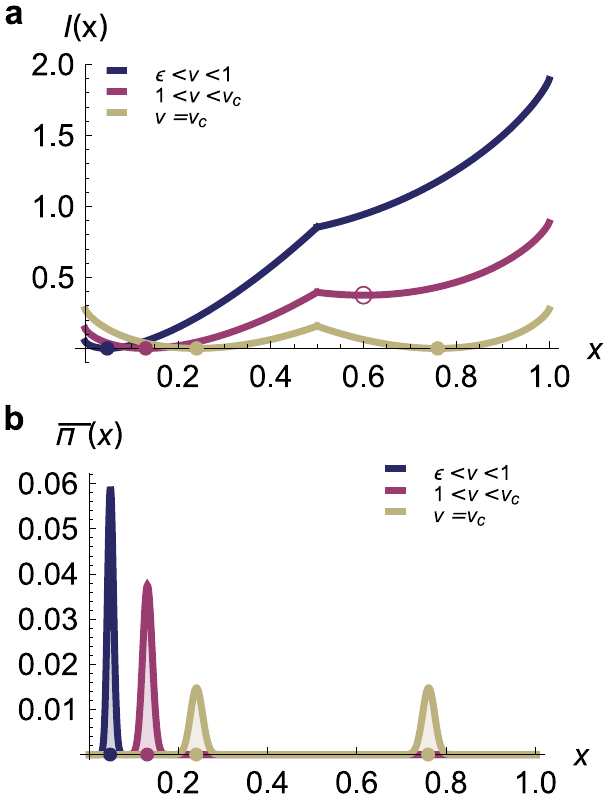}
\caption{The steady state distribution is unimodal when $v\ne v_c$ and bimodal when $v=v_c$. (\textbf{a}) The entropy function $I(x)=J(x)-J_0$ for different values of $v$ and for $\varepsilon=0.1$. The minima are highlighted with dots. When $v=v_c=1/\sqrt{\varepsilon}$, the entropy function possesses two global minima.
(\textbf{b}) The measure concentrates at the global minima (plain dots) of the entropy function. This is illustrated for $N=1000$.
}\label{fig:EntropyFunction}
\end{figure}
The limiting behavior of the steady state distribution $\bar\pi_N$ is
\begin{equation}\label{eq:limit:pi:enciso}
\lim_{N\to\infty}\bar\pi_{N}(x)= \begin{cases}
\delta_{\frac{\varepsilon v}{1+\varepsilon v}}(x), & \text{if } v <v_c\,, \\
\frac{1}{2}\,\delta_{\frac{\varepsilon v}{1+\varepsilon v}}(x)+\frac{1}{2} \, \delta_{\frac{ v}{1+ v}}(x), & \text{if } v =v_c, \\
\delta_{\frac{ v}{1+ v}}(x), & \text{if } v>v_c.
\end{cases}
\end{equation}

\subsubsection{Hill coefficients}

As seen in (\ref{Mixture0}), the steady state $\bar\pi_N$ is a mixture of the binomial distributions $\pi_1=\mathcal{B}(N,\frac{\varepsilon v}{1+\varepsilon v})$ and $\pi_2=\mathcal{B}(N,\frac{v}{1+v})$, that is,
$$\bar\pi_N=\alpha_N(v) \pi_1 + (1-\alpha_N(v))\pi_2$$
for the coefficient
$$\alpha_N(v)=\frac{1}{1+\lp\frac{\sqrt{\varepsilon}(1+ v)}{1+\varepsilon v}\rp^{N}}\,.$$
When $a(x)\equiv x$, one obtains that
$$f(v)=\alpha(v)\frac{\varepsilon v}{1+\varepsilon v}+(1-\alpha(v))\frac{v}{1+v} .$$

Consider the quantile $v_q^{(N)}$ given by the equation $q=f(v_q^{(N)})$. We show in the Appendix that the following result for $I_q$ holds.
\begin{lemma}\label{allosteric:Ip}
Assume that $a(x) \equiv x$ and let $q$ be such that $1/2\leq q \leq 1/(1+\sqrt{\varepsilon})$. Then
\begin{eqnarray}\label{eq:limvpN}
\lim_{N\to\infty}v_q^{(N)}= \lim_{N\to\infty}v_{1-q}^{(N)} = v_c
\end{eqnarray}
and therefore
$$\lim_{N\to\infty}I_q = +\infty.$$
\end{lemma}
Fig. \ref{fig:koshland:enciso} (b) and (c) shows that  $I_q$ is asymptotically linear in $N$.
On the other hand, for $q$ close to $1$, then $I_q$ is asymptotically constant, so that the related
index does not detect ultrasensitivity of order $N$, see Fig.~\ref{fig:koshland:enciso}. This shows that experimentalists should use a broad range of 
coefficient $I_q$ instead of focusing only on $I_{0.9}$. 

The Hill coefficient $\eta_H(v)$ possesses nice thermodynamical interpretations, see, e.g., formulas
(\ref{Formula1}) and (\ref{Formula2}), and is such that $\eta_H(v)\sim C_v$ for constants $C_v$ when $v\ne v_c$, while
$\eta_H(v_c)\sim C_{v_c} N$, when $v=v_c$. On the other hand, the effective Hill coefficient $I_q$, which is nothing but an average of Hill coefficients (see (\ref{Relation})) diverges as $N\to\infty$ for a broad range of values of $q$. Detection of ultrasensitivity should be therefore easier with $I_q$ than with $\eta_H(v)$.

\begin{figure}
\centering
\includegraphics{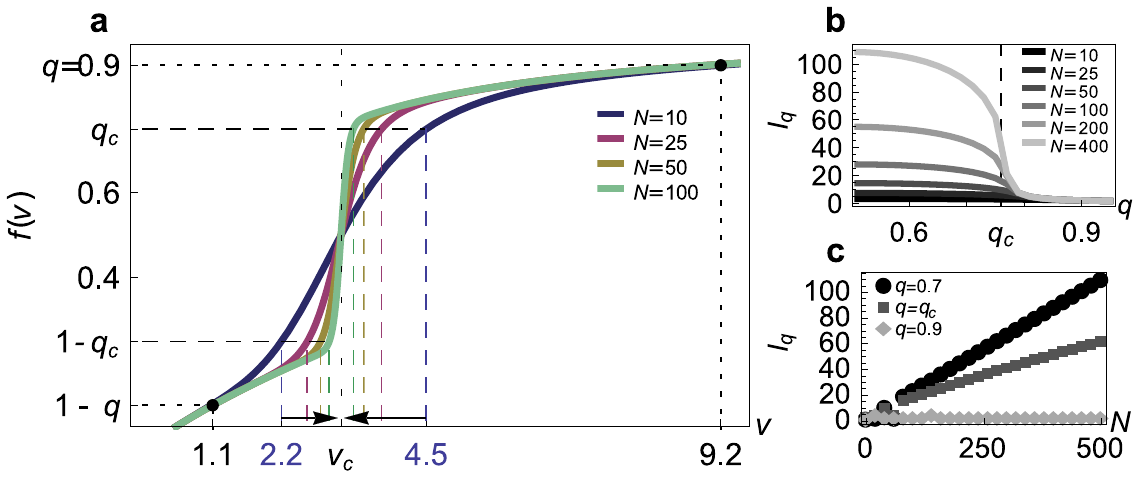}
\caption{Numerical approximation of the effective Hill coefficient. The activity function is $a(x)=x$ and the parameter $\varepsilon=0.1$. In this case, $v_c=\sqrt{\varepsilon}^{-1}\approx 3.2$ and $x_2\approx 0.76$.
(\textbf{a}) The activity of the macromolecule $f(v)$ for different values of $N$.  When $q=0.9$, $v_{0.9}$ and $v_{0.1}$ converge to different limits in the large $N$ limit. The Koshland index is of order $N$ when $q$ is smaller than the threshold $q_c = x_2(v_c)$, and $v_q-v_{1-q}$ converges towards zero so that $I_q$ goes to infinity.
(\textbf{b}) The effective Hill coefficient $I_q$ as a function of $q$ for several values of $N$.
(\textbf{c}) The index $I_q$ as a function of $N$, for three values of $q$. We see that the index $I_q$, in broad outline, takes two values. The first one goes to infinity with $N$ and the second one goes to 1, for larger $q$.}\label{fig:koshland:enciso}
\end{figure}

\subsection{Substrate-Catalyst interactions \label{s.substrate}}

\begin{figure}
\begin{center}
\includegraphics{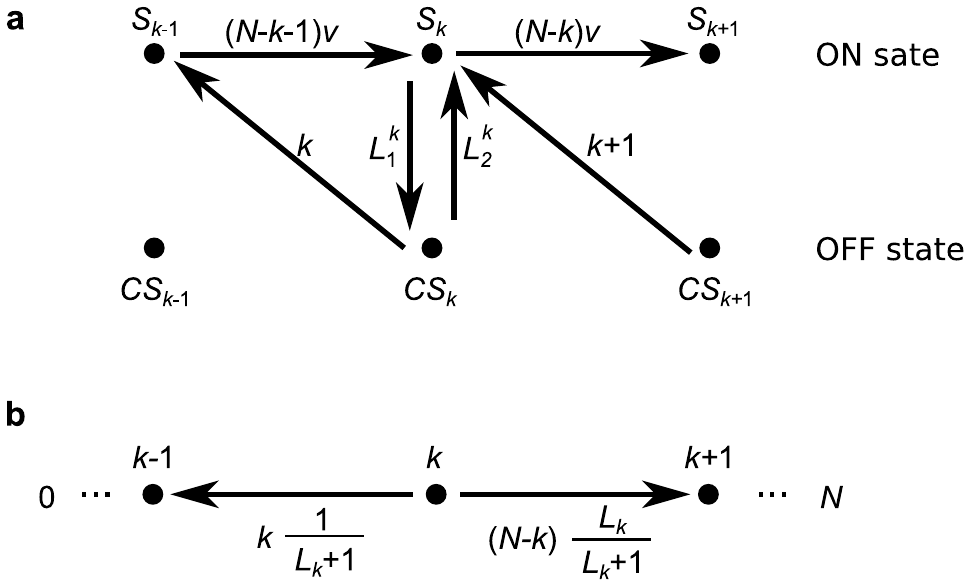}
\end{center}
\caption{Substrate-Catalyst reaction. (\textbf{a}) Transition scheme associated with substrate-catalyst reactions. (\textbf{b}) Associated birth and death process.}\label{fig:stripKaneko}
\end{figure}

Consider a substrate molecule containing $N$ sites where ligand molecules can bind at rate $v$. The transition rates are provided in Fig.~\ref{fig:stripKaneko}.
This model has been used for example in~\cite{Hatakeyama:2014aa} for proposing a mechanism depicting kinetic memory. The substrate-catalyst reactions
\begin{equation*}
 {\rm S}_k + {\rm C}\underset{L_2^k }{\overset{ L_1^k}{\longleftrightarrow}}{\rm S_k C} {\overset{ k}{\longrightarrow}} {\rm S}_{k-1},
\end{equation*}
and the reactions associated with ligand binding
\begin{equation*}
 {\rm S}_k {\overset{ v(N-k)}{\longleftrightarrow}}{\rm S}_{k+1},
\end{equation*}
define a Markov chain evolving on a strip, see Fig.~\ref{fig:stripKaneko} (a).
The catalytic reactions are simplified by assuming fast formation of the complex $CS_k$, by supposing that $L_1^k, \, L_2^k \gg1$ for every $0\leq k \leq N$. This last hypothesis leads to an approximating birth and death process, see below and Fig.~\ref{fig:stripKaneko} (b).
We follow \cite{Hatakeyama:2014aa} by setting
$$L_k = \frac{L_1^k}{L_2^k} = \frac{L_1^0}{L_2^0}\gamma^{\frac{k}{N}} = L_0 \gamma^{\frac{k}{N}}$$
for some positive constant $\gamma$, see Fig.~\ref{fig:stripKaneko}. These two sets of reactions define a birth and death process $Y_N(t)$ of transition rates
$$
q_N\lp k,\,k+1 \rp = v(N-k)\frac{L_k}{L_k+1} \quad\hbox{ and }\quad
q_N\lp k,\,k-1 \rp = k\frac{1}{L_k + 1}.$$
and of steady state
\begin{equation}\label{eq:stationary:kaneko}
\bar\pi_{N}(k) =\lp v L_0 \rp^k {N \choose k} \gamma^{\frac{(k-1)k}{2N}} \frac{L_0 \gamma^{\frac{k}{N}}+1}{L_0 +1} \bar\pi_{N}(0),
\end{equation}
where $\bar\pi_{N}(0) = 1/\sum_{k=0}^N\lp v L_0 \rp^k {N \choose k} \gamma^{\frac{(k-1)k}{2N}}\frac{L_0 \gamma^{\frac{k}{N}}+1}{L_0 +1}$ is the inverse of the normalizing constant.

The limiting (o.d.e.) (\ref{ODE}) is given by
\begin{equation}\label{ODE2}
\frac{\textrm{d}x}{\textrm{d}t} = \frac{v(1-x)L_0\gamma^x - x}{L_0\gamma^x + 1}
\end{equation}
which can possess multiple stable equilibria. 

\subsubsection{Entropy function and concentration of measure}

\begin{figure}\centering
\includegraphics{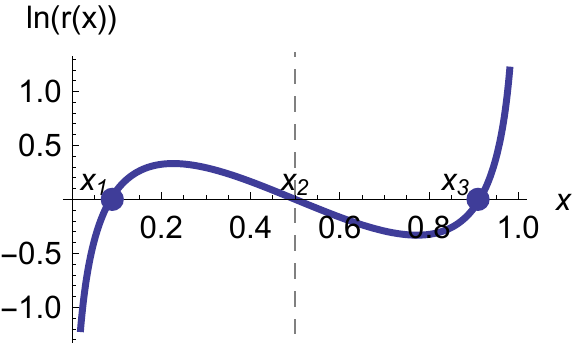}
\caption{Plot of the function $\ln(r(x))$. This function can possess three roots $x_1$, $x_2$ and $x_3$ and $\ln(r(x+x_2))$ is odd. The roots $x_1$ and $x_3$ (dots) are the stable equilibria of $I(x)$.}
\label{fig:largeDev:kaneko}
\end{figure}

In this setting, 
$$\ln\lp r(x)\rp = \ln\lp\frac{x}{1-x}\rp + \ln\lp\frac{1}{vL_0}\rp - \ln\lp\gamma^x\rp,$$
see Fig.~\ref{fig:largeDev:kaneko}.
The main difference with the previous examples comes from this third term. Notice that when $\gamma=1$, we recover the binomial distribution of the introducing example. We assume that $\gamma > 1$ in what follows and give conditions describing the possible critical points of the related action functional $J$. This result is proven in the Appendix and the overall picture is summarized in Fig.~\ref{fig:kaneko:hill}. Define,
\begin{align}\label{eq:kaneko:cond}
(C1) \quad & \ln\lp\gamma\rp > 4 \\
(C2) \quad & \frac{1}{2}-\frac{1}{2} \sqrt{\frac{\ln\lp\gamma\rp-4}{\ln\lp\gamma\rp}}<x<\frac{1}{2}+\frac{1}{2} \sqrt{\frac{\ln\lp\gamma\rp-4}{\ln\lp\gamma\rp}}\\
(C3) \quad & \gamma^{-\frac{1}{2} \left(1+\sqrt{1-\frac{4}{\ln(\gamma)}}\right)}c_{\gamma}> vL_0 > \gamma^{-\frac{1}{2} \left(1-\sqrt{1-\frac{4}{\ln(\gamma)}}\right)} \frac{1}{c_{\gamma}} \\
(C4) \quad & v = v_c = \frac{1}{\sqrt{\gamma}L_0},
\end{align}
where 
$$c_{\gamma} = \frac{\left(\sqrt{\ln(\gamma)-4}+\sqrt{\ln(\gamma)}\right)}{\sqrt{\ln(\gamma)}-\sqrt{\ln(\gamma)-4}}.$$

\begin{lemma}\label{lemma:kaneko:largeDev}
Under conditions $(C1)$-$(C3)$, the entropy function $I(x)$ possesses two local minima. Furthermore, these two minima are global when $(C4)$ holds.
\end{lemma}

The effective Hill coefficient is represented in Fig.~\ref{fig:kaneko:hill}, where the same conclusions as in the previous example hold.


\begin{figure}
\centering
\includegraphics{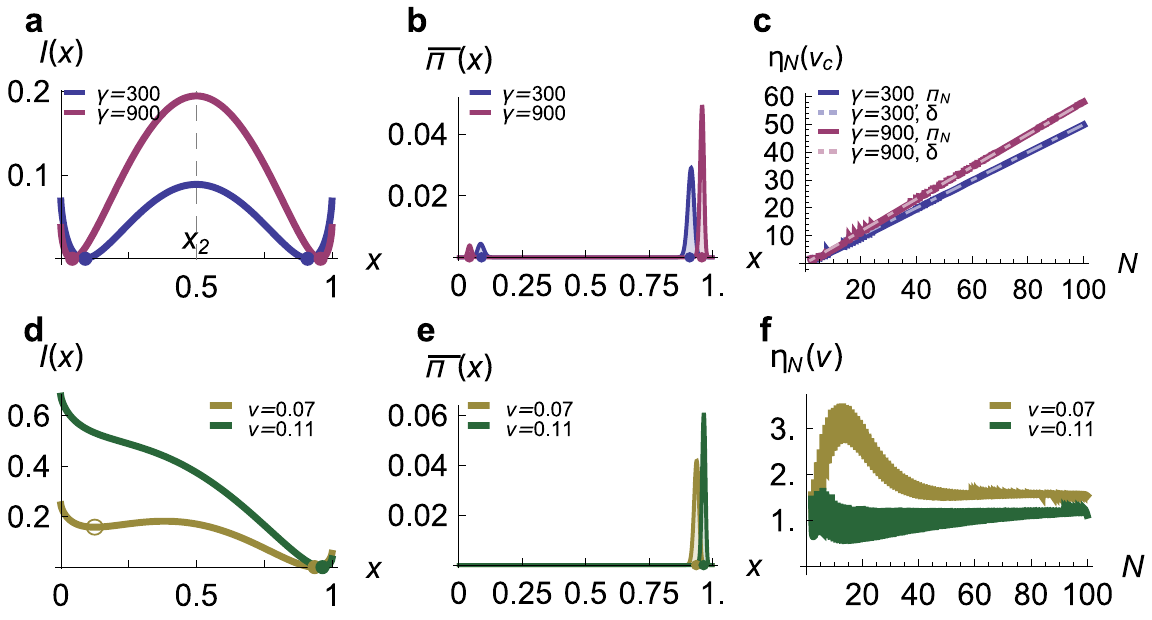}
\caption{System's bistability and ultrasensitivity. $a(x) = x$. In the first line $v=v_c$ and $\gamma$ varies. In the second line $\gamma=300$ is fixed, but $v\neq v_c$. (\textbf{a}) The function $I(x)$ possesses two global minima $x_1$ and $x_3$. $x_2$ is the unstable equilibrium of the (o.d.e.) (\ref{ODE2}). (\textbf{b}) The steady state converges towards a mixture of Dirac masses at $x_1$ and $x_3$. Here, $N=1000$. (\textbf{c}) The Hill coefficient is of order $N$ (continuous lines) and follows the same trend as the Hill coefficient associated with the limiting stationary distribution (dotted lines). (\textbf{d}) When $v\neq v_c$, $I(x)$ possesses potentially two minima, but only one of them (plain disks) is global. (\textbf{e}) The stationary distribution is unimodal ($N=1000$). (\textbf{f}) The system does not exhibit ultrasensitivity.}\label{fig:kaneko:hill}
\end{figure}

\begin{figure}
\centering
\includegraphics{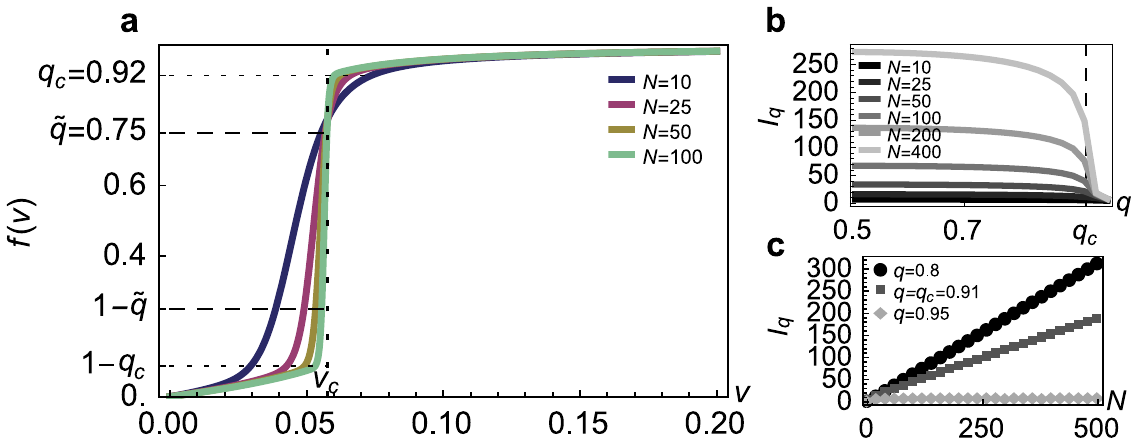}
\caption{Numerical approximation of the effective Hill coefficient. The activity function is $a(x)=x$ and the parameter $\gamma=300$, $L_0=1$. In this case $v_c \approx 0.057$ and $x_2\approx 0.92$.
(\textbf{a}) The activity of the macromolecule $f(v)$ for different values of $N$. Ultrasensitivity is detected for the standard $q=0.9$ since $0.9<x_2(v_c) = q_c$ and $I_q$ goes to infinity.
(\textbf{b}) The effective Hill coefficient $I_q$ as a function of $q$ for several values of $N$. (\textbf{c}) The index $I_q$ as a function of $N$, for three values of $q$.}\label{fig:koshland:kaneko}\label{fig:kaneko:koshland}
\end{figure}

\subsection{Nucleosome mediated cooperativity\label{s.nucleosome}}

This section presents the model of~\cite{mirny2010} which shows how indirect cooperativity can result from the competition between nucleosome positioning on the DNA and hampering transcription factors (TF) trying to occupy free sites on the DNA, see Fig.~\ref{fig:mirny:band}. This fact has been demonstrated in~\cite{mirny2010} and~\cite{MazzaBenaim}. We will show here how one can understand this phenomenon using our general framework.

Let us define $W(t)$ such that $W(t)=A$ if the nucleosome is not bound to the DNA at time $t$ (active state) and $W(t)=I$ otherwise (inactive state). The DNA possesses $N$ binding sites and the nucleosome can access or leave the DNA only when all binding sites are free of (TF). The transitions between active and inactive state occur at rate $g$ (inactive to active) and $\kappa$ (active to inactive), such that $L = g/\kappa = K^N$, see Fig.~\ref{fig:mirny:band}.

Let $N(t)$ denotes the number of occupied sites at time $t$. The birth rate is given by $\mu_A v$ in the active state and by $\mu_I v <\mu_A v$ in the inactive state. Here $v$ denotes the protein concentration. Fig.~\ref{fig:mirny:band} sums up the whole process, which is equivalent to the Monod-Wyman-Changeux model of allosteric regulation~\cite{monod1965}, see, e.g., \cite{Edelstein}.

\begin{figure}[h!]
\begin{center}
\includegraphics{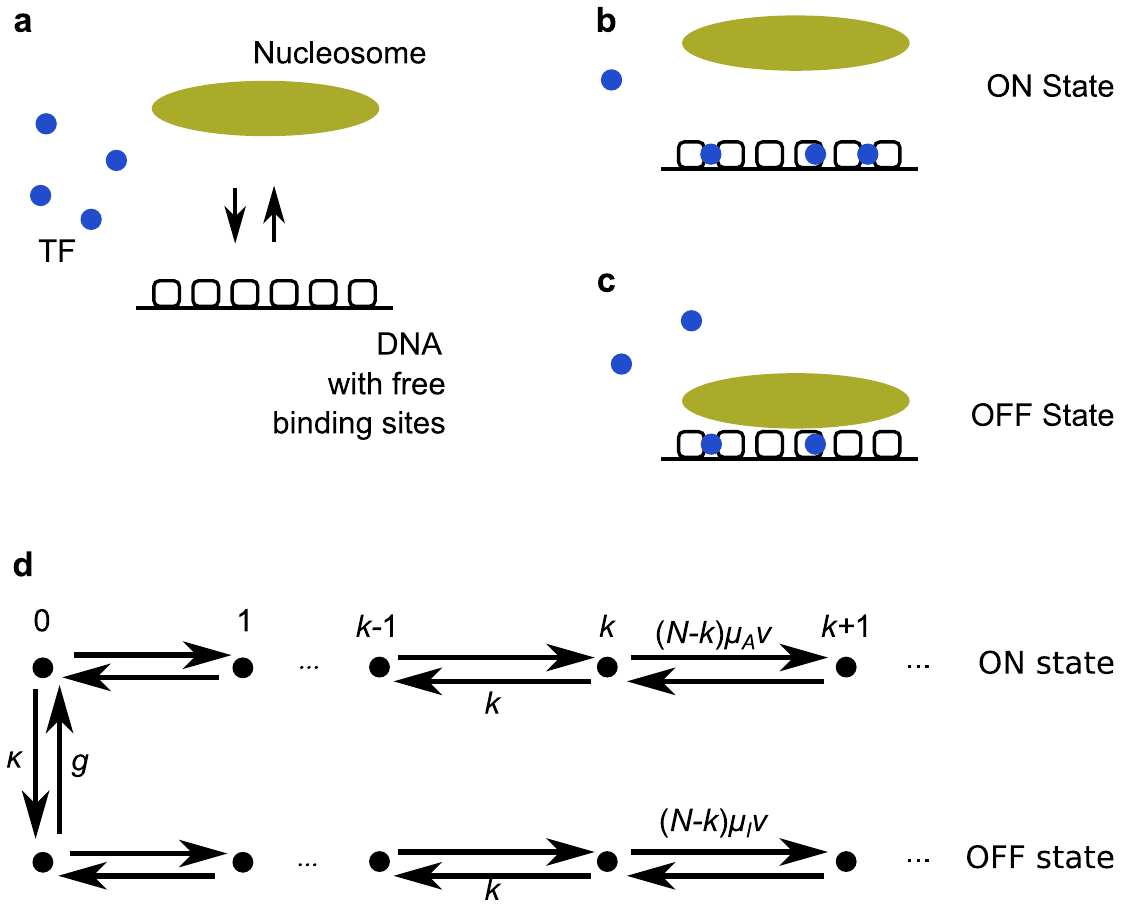}
\end{center}
\caption{Nucleosome mediated cooperativity. (\textbf{a}) Transcription factors (TF) try to access free binding sites on the DNA. The nucleosome can bind/unbind only when all binding sites are free from (TF). (\textbf{b}) In the active (on) state, the nucleosome is unbound and (TF) can easily access the DNA. (\textbf{c}) In the inactive (off) state, the nucleosome is bound and hampers the access of (TF). (\textbf{d}) The process can switch from the active (on) to the inactive (off) state only when $N(t)=0$. In each state, the process is a birth and death process.}\label{fig:mirny:band}
\end{figure}

Following~\cite{mirny2010} and~\cite{MazzaBenaim} the full Markov-chain process $\lp N(t), W(t)\rp$ at equilibrium is such that
\begin{eqnarray*}
\Prob\lp N(\infty) = k \, ,\, W(\infty) = I \rp &=& \frac{{N\choose{k}} \lp \mu_I v \rp^k }{\nu_I + L \nu_A},\\
\Prob\lp N(\infty) = k \, ,\, W(\infty) = A \rp &=& \frac{L{N\choose{k}} \lp \mu_A v \rp^k }{\nu_I + L \nu_A},
\end{eqnarray*}
where $\nu_I$ and $\nu_A$ are
$$\nu_I = \lp 1 + \mu_Iv \rp^N \quad \text{and} \quad \nu_A = \lp 1 + \mu_Av \rp^N.$$
This measure is the invariant measure of a birth and death process evolving in a segment composed of three pieces (see~\cite{MazzaBenaim} and Fig.~\ref{fig:mirny:band}). A computation shows that
\begin{eqnarray*}
\Prob\lp W(\infty) = I \mid N(\infty) = k \rp &=& \frac{L\lp \mu_I v \rp^k}{L\lp \mu_I v \rp^k + \lp \mu_A v \rp^k},\\
\Prob\lp W(\infty) = A \mid N(\infty) = k \rp &=& \frac{\lp \mu_A v \rp^k}{L\lp \mu_I v \rp^k + \lp \mu_A v \rp^k}.
\end{eqnarray*}
This leads to a birth and death process $Y_N(t)$ counting the number of bound sites in either the active or the inactive state, of birth and death rates
\begin{equation*}
q_N\lp k, k+1 \rp = \lp N-k \rp v \, \frac{\mu_I^{k+1} + L \mu_A^{k+1}}{\mu_I^{k} + L \mu_A^{k}} \quad\text{ and }\quad 
q_N\lp k, k-1 \rp = k.
\end{equation*}
The stationary law of this chain is
$$\bar \pi_N(k) =\frac{{N\choose k} \lp \lp \mu_I v\rp^k + L \lp \mu_1 v \rp^k \rp}{\nu_I + L \nu_A}.$$
When $N\to \infty$, recalling that $L=K^N$ and defining $x=k/N$, the above chain rescaled on $\left[0,1\right]$ is such that
\begin{eqnarray*}
b(x) &=& \lim_{N\to\infty} (1-x) \mu_I v \frac{1+\lp K \lp \frac{\mu_A}{\mu_I} \rp^x \rp^N \, \frac{\mu_A}{\mu_I}}{1+\lp K \lp \frac{\mu_A}{\mu_I} \rp^x \rp^N}\\
&=& \begin{cases}
(1-x) \mu_A v 					& \text{ if } x > x_c,\\
(1-x) \frac{\mu_A + \mu_I}{2} v	& \text{ if } x = x_c,\\	
(1-x) \mu_I v 					& \text{ if } x < x_c,
\end{cases}
\end{eqnarray*}
and $d(x) = x$. Notice that the critical value $x_c=\frac{\ln(1/K)}{\ln(\mu_A/\mu_I)}$ belongs to the interval $\left[0,1\right]$ if and only if 
\begin{eqnarray}\label{eq:mirny:cond:deuxZeros}
1 > K > \frac{\mu_I}{\mu_A}.
\end{eqnarray}

This leads to the limiting (o.d.e.)(\ref{ODE})
\begin{equation}\label{ODE3}
\frac{\textrm{d}x}{\textrm{d}t} = 
\begin{cases}
(1-x)\mu_A v -x, & \text{ for } x >x_c,\\
(1-x)\mu_I v - x, & \text{ for } x <x_c,
\end{cases}
\end{equation}
which possesses two stable equilibria $x_1 = \frac{v\mu_I}{1+v\mu_I}$ and $x_2 = \frac{v\mu_A}{1+v\mu_A}$ in $\left[0,1\right]$ when condition (\ref{eq:mirny:cond:deuxZeros}) is satisfied.

\subsubsection{Entropy and concentration of measure}

In the present framework,
$$\ln(r(x)) = \begin{cases} 
\ln\lp \frac{x}{1-x}\rp + \ln\lp \frac{1}{v \mu_A}\rp, 			& \text{ if } x>x_c,\\
\ln\left(\frac{2}{(\mu_I+\mu_A) v}\right)+ \ln\lp \frac{-\ln\lp\frac{1}{K} \rp}{\ln\lp\frac{1}{K}\rp - \ln\lp\frac{\mu_A}{\mu_I}\rp }\rp, & \text{ if } x= x_c,\\
\ln\lp \frac{x}{1-x}\rp + \ln\lp \frac{1}{v \mu_I}\rp ,				& \text{ if } x< x_c,
\end{cases}$$
which admits two zeros $x_1$ and $x_2$ in $\left[0,1 \right]$, which are the two stable equilibria of the limiting (o.d.e.) (\ref{ODE3}) and also the critical points of the action functional
\begin{align}\label{eq:mirny:j}
J(x) & = \int_0^x \ln\lp r(y) \rp dy \nonumber \\
& = x\ln(x) + (1-x)\ln(1-x) - \lp x\ln\lp v\mu_A\rp + \log\lp K \rp\rp  \ind_{\lac x > x_c\rac} -  x\ln\lp v\mu_I \rp \ind_{\lac x \leq x_c\rac}.
\end{align}
Direct computations show that when (\ref{eq:mirny:cond:deuxZeros}) holds and for 
$$v=v_c=\frac{1-K}{K\mu_A - \mu_I},$$ 
the limiting stationary distribution of the birth and death process converges to a probability measure $\bar \pi$ which is a mixture of Dirac measures concentrated at the equilibria $x_1$ and $x_2$. Fig.~\ref{fig:mirny:hill} provides two illustrations when $v\ne v_c$ and $v=v_c$, the latter situation leading to a bimodal steady state with $\eta_H(v_c)$ linear in the number of binding sites $N$. In Fig.~\ref{fig:mirny:koshland}, one sees that for an appropriate choice of $q$, the effective Hill coefficient $I_q$ detects the switch-like response, as it is proportional to the number of binding sites.

\begin{figure}
\centering
\includegraphics{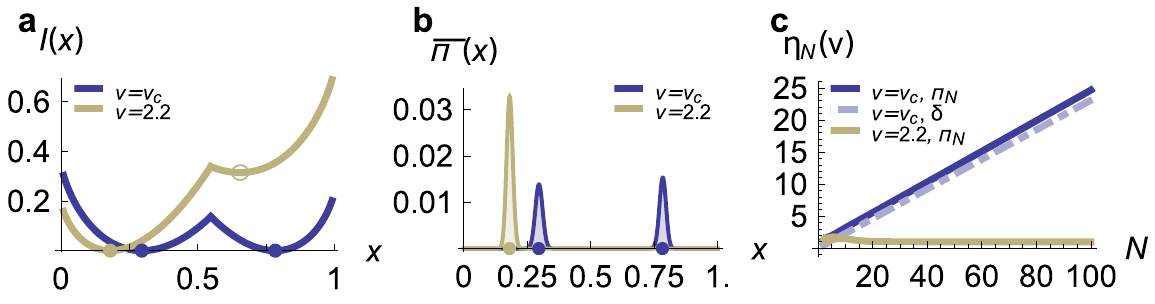}
\caption{System's bistability and ultrasensitivity. Condition (\ref{eq:mirny:cond:deuxZeros}) is satisfied with $\mu_I = 0.1$, $\mu_A = 0.86$ and $K=0.3$. (\textbf{a}) The entropy function $I(x)$ possesses two roots when $v=v_c$. (\textbf{b}) The stationary distribution is bimodal when $v=v_c$ (illustrated here for $N=1000$). (\textbf{c}) When $v=v_c$, the limiting stationary distribution possesses thus two modes situated at the rescaled points $x_1$ and $x_2$ corresponding to the zeros of $I(x)$. The Hill coefficient is hence proportional to $N$ and the system is ultrasensitive.}\label{fig:mirny:hill}
\end{figure}

\begin{figure}
\centering
\includegraphics{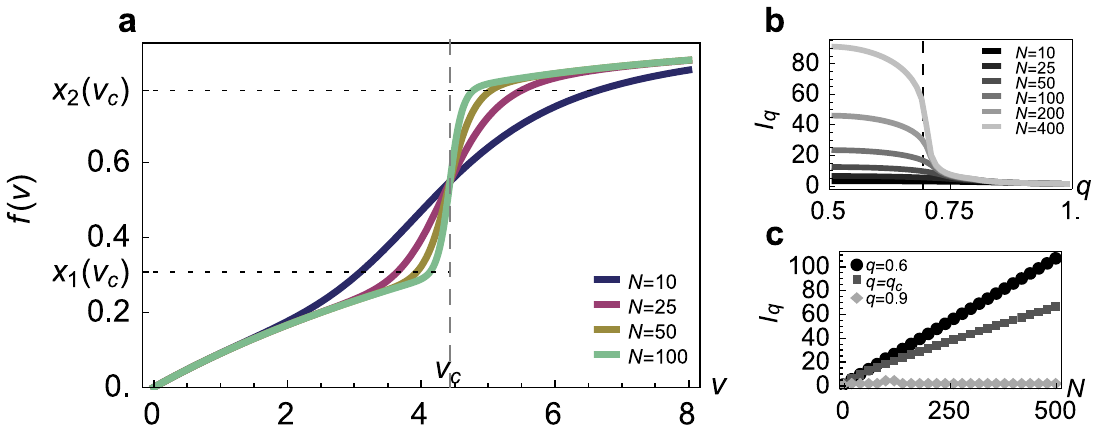}
\caption{Numerical approximation of the effective Hill coefficient. (\textbf{a}) The activity of the macromolecule $f(v)$ for different values of $N$ and $a(x)\equiv x$. The parameters are given by $\mu_I = 0.1$, $\mu_A = 0.86$ and $K=0.3$. When $q< \min\left\{x_2(v_c),\, 1-x_1(v_c) \right\}$, the effective Hill coefficient $I_q$ goes to infinity when $N\to \infty$. Otherwise $I_q$ converges towards a constant. (\textbf{b}) The effective Hill coefficient index as a function of $q$ for several values of $N$. (\textbf{c}) The index $I_q$ as a function of $N$ for three values of $q$.}\label{fig:mirny:koshland}
\end{figure}



\appendix

\section{Proof of Theorem~\ref{SecondFormula}}
In this section, the proofs of new formulas for Hill coefficients are detailed. 
These formulas follow in fact from a more general result allowing to compute the Hill coefficient for steady-state distribution that are mixture of measures of the form~\eqref{BasicProbability}. 


\begin{lemma}\label{Mixture}
Let $\pi_1$ and $\pi_2$ be two probability measures of the form given in~\eqref{BasicProbability}. For
any differentiable function $0\le\alpha(v)\le 1$, consider the probability mixture $\pi_\alpha$ given by
$\pi_\alpha = \alpha \pi_1 +(1-\alpha)\pi_2$, and consider the dose response curve
$f(v)= \langle a(\frac{\vert n\vert}{N})\rangle_{\pi_\alpha}$. Then, the Hill coefficient is given by
\begin{align*}
v\frac{\,{\rm d}}{\,{\rm d}v}\ln\Big(\frac{f(v)}{f(\infty)-f(v)}\Big)=\Big( &N\alpha(v){\rm Cov}_{\pi_1}\lp a(\tfrac{\vert n\vert}{N}),\tfrac{\vert n\vert}{N}\rp 
+N(1-\alpha(v)){\rm Cov}_{\pi_2}\lp a(\tfrac{\vert n\vert}{N}),\tfrac{\vert n\vert}{N}\rp \\
 & + v\frac{\,{\rm d}\alpha(v)}{\,{\rm d}v}(\langle a(\tfrac{\vert n\vert}{N})\rangle_{\pi_1}-\langle a(\tfrac{\vert n\vert}{N})\rangle_{\pi_2})\Big)
 \frac{f(\infty)}{f(v)(f(\infty)-f(v))},
 \end{align*}
 where $f(\infty)=a(\frac{k^*}{N})$.
\end{lemma}

\proof
Recall that by~\eqref{BasicProbability},
$$\pi_{\alpha}(n) = \alpha(v) \cdot \frac{\mu_1(n) v^{|n|}}{Z_1(v)} + \lp1-\alpha(v)\rp \cdot \frac{\mu_2(n) v^{|n|}}{Z_2(v)},$$
where $\mu_i$ and $Z_i$ are the non-negative weights, respectively the normalization constants of the measures $\pi_i$ with $i \in\left\{1,2 \right\}$.
We compute thus
$$v\frac{\,{\rm d}}{\,{\rm d}v}\ln\Big(\frac{f(v)}{f(\infty)-f(v)}\Big) = v\cdot \frac{f'(v) f(\infty)}{f(v)\lp f(\infty) - f(v) \rp}.$$
An explicit computation of $f'(v) = \sum_{n \in \Lambda} a(\tfrac{|n|}{N})\cdot\frac{\,{\rm d}}{\,{\rm d}v}\pi_{\alpha}(n)$ leads to the result.
\endproof


%

\begin{proof}[Proof of Theorem~\ref{SecondFormula}]
The first part of the theorem follows directly from Lemma~\ref{Mixture} by considering a mixture coefficient $\alpha(v)\equiv1$. For the asymptotic behavior of $\eta_H(v)$, notice that $\lim_{v\to\infty} \langle a( \tfrac{|n|}{N}) \rangle_{\pi} = a(\tfrac{k^*}{N})$. We compute
\begin{eqnarray*}
\lim_{v\to \infty} \frac{{\rm Cov}_{\pi}\lp a(\tfrac{|n|}{N}),\tfrac{|n|}{N} \rp}{a(\tfrac{k^*}{N}) - \langle a(\tfrac{|n|}{N}) \rangle_{\pi}} 
& = & \frac{1}{N} \lim_{v\to\infty} \frac{\frac{\sum_k k a(\tfrac{k}{N}) V(k) v^k}{Z(v)} - \frac{\sum_k k V(k) v^k\sum_k a(\tfrac{k}{N})V(k) v^k}{(Z(v))^2} }{a(\tfrac{k^*}{N}) - \frac{\sum_k a(\tfrac{k}{N}) V(k) v^k}{Z(v)}}.
\end{eqnarray*} 
We will compute each division and keep track of the leading coefficients in the resulting polynomials,
\begin{align*}
&\frac{\sum_k k a(\tfrac{k}{N})V(k) v^k}{Z(v)} = k^*a(\tfrac{k^*}{N}) + \frac{V(k^*-1)\lp (k^*-1)a(\tfrac{k^*-1}{N}) - k^* a(\tfrac{k^*}{N}) \rp v^{k^*-1} + ...}{V(k^*)v^{k^*} + V(k^*-1)v^{k^*-1} + ...},\\
& \frac{\sum k V(k) v^k\sum a(\tfrac{k}{N})V(k) v^k}{(Z(v))^2}  =  k^*a(\tfrac{k^*}{N}) \\
 &\qquad +  \frac{V(k^*)V(k^*-1)\lp k^*a(\tfrac{k^*-1}{N}) - a(\tfrac{k^*}{N}) (1 + k^*)\rp v^{2k^*-1} + ...}{(V(k^*))^2v^{2k^*} + 2V(k^*)V(k^*-1)v^{2k^*-1} + ...},\\
&\frac{\sum_k a(\tfrac{k}{N}) V(k) v^k}{Z(v)} = a(\tfrac{k^*}{N}) + \frac{\lp a(\tfrac{k^*-1}{N}) - a(\tfrac{k^*}{N}) \rp V(k^*-1)v^{k^*-1} + ...}{V(k^*)v^{k^*} + ...}.
\end{align*}
All together, this gives
$$\lim_{v\to \infty} \frac{{\rm Cov}_{\pi}\lp a(\tfrac{|n|}{N}),\tfrac{|n|}{N} \rp}{a(\tfrac{k^*}{N}) - \langle a(\tfrac{|n|}{N}) \rangle_{\pi}} = \frac{1}{N},$$
so that $\lim_{v\to\infty}\eta_H(v) = 1$. The limit when $v\to 0$ is computed in the same way.
\end{proof}

\section{Proof of Theorem~\ref{AsymptoticHill}}

\begin{proof}[Proof of Theorem~\ref{AsymptoticHill}]
By Theorem~\ref{SecondFormula}, we have
$$\lim_{N \to \infty}\frac{\eta_H(v)}{N} =\lim_{N\to \infty} \frac{\textrm{Cov}_{\bar\pi_N}(X_N,a(X_N))a(\tfrac{k^*}{N})}{f(v) (a(\tfrac{k^*}{N}) - f(v))}.$$
Assume first $m=1$, that is the entropy $I(x)$ has a unique equilibrium $0<x_1<1$. In this case, it was already established in~\cite{chan1998largedeviation} that $\bar\pi_N$ converges weakly to a Dirac mass $\delta_{x_1}$. Therefore, $f(v) = \langle a(X_N) \rangle_{\bar\pi_N} \xrightarrow[N\to\infty]{} a(x_1)$. Since by assumption $a$ is strictly increasing and $0<x_1<\lim \frac{k^*}{N}$, the denominator of the above expression is well defined. Using in addition the assumption of continuity and boundedness of $a$, we obtain $$\textrm{Cov}_{\bar\pi_N}(a(X_N),\,X_N) \xrightarrow[N\to\infty]{} \textrm{Cov}_{\delta_{x_1}}(a(X_{\infty}),\,X_{\infty}) = 0,$$ 
where $X_{\infty}$ denotes the limiting process of $X_N$ whose degenerate steady state is $\delta_{x_1}$. This establishes the first part of the statement.\\

Assume now the existence of $m>1$ different equilibria with $x_1<x_2<\ldots<x_m$ and suppose for convenience that $k^*=N$. By Corollary~\ref{cor:WeakConvergence}, we know that $\bar\pi_N$ converges weakly to $\sum_{i=1}^m c_i \delta_{x_i}$ for some positive constants $c_1,\ldots, c_m$ satisfying $\sum_{i=1}^m c_i = 1$. Continuity and boundedness of the function $a$ lead to 
\begin{eqnarray*}
\lim_{N\to \infty}\textrm{Cov}_{\bar\pi_N}\lp a(X_N), X_N\rp & = & \textrm{Cov}_{\sum_{i=1}^m c_i \delta_{x_i}}\lp a(X_{\infty}),\, X_{\infty}\rp\\
& = & \sum_{i=1}^m c_i x_i a(x_i) - \sum_{i=1}^m c_i x_i \sum_{j=1}^m c_j a(x_j)
\end{eqnarray*}
and the right hand term is positive. Indeed, 
\begin{eqnarray*}
\sum_{i=1}^m c_i x_i a(x_i) & = & \sum_{j=1}^m c_j \sum_{i=1}^m c_i x_i a(x_i) \\
& = & \sum_{i=1}^m c_i^2 x_i a(x_i ) + \sum_{j=1}^m \sum_{i>j}c_j c_i x_i a(x_i) + \sum_{j=1}^m \sum_{i<j}c_j c_i x_i a(x_i)\\
& = & \sum_{i=1}^m c_i^2 x_i a(x_i ) + \sum_{k>l}c_k c_l \lp x_l a(x_l) + x_k a(x_k) \rp
\end{eqnarray*}
and $x_l a(x_l) + x_k a(x_k) > x_k a(x_l) + x_l a(x_k)$ for $k>l$, since $x_l<x_k$ and $a(x_l)<a(x_k)$. Therefore,
\begin{eqnarray*}
\sum_{i=1}^m c_i x_i a(x_i) 
& > & \sum_{i=1}^m c_i^2 x_i a(x_i ) + \sum_{k>l}c_k c_l \lp x_k a(x_l) + x_l a(x_k)\rp \\
& = & \sum_{i=1}^m c_i x_i \sum_{j=1}^m c_j a(x_j),
\end{eqnarray*}
which concludes the proof.
\end{proof}

\subsection{Technical lemmas} \label{sec:technical:lemmas} 
In this part, we compute precise large deviations of the steady state measure $\bar\pi_N$ in order to establish that it charges all global minima of the entropy function $I(x)$ (see Lemma~\ref{lemma:pi:charges:equilibr}) so that it converges weakly to a combination of Dirac measures (see Corollary~\ref{cor:WeakConvergence}). Lemmas~\ref{lemma:estimates:e1}, \ref{lemma:pi(j/N)} and \ref{lemma:pi(0):J(0)} are intermediate results needed for the proof.

\begin{lemma}\label{lemma:pi:charges:equilibr}
Suppose that \eqref{Condition1} and Assumptions~\ref{assump:r} are satisfied. Then, the probability measure $\bar\pi_N$ charges all equilibria $x_i$, $i=1,\ldots, m$. That is, for all $i=1,\ldots,m$, 
\begin{equation}\label{liminf}
\liminf_{N\to\infty} \bar\pi_N((x_i-\varepsilon,x_i+\varepsilon))>0,\qquad\text{ for any }\varepsilon > 0.
\end{equation}
\end{lemma}
\begin{proof}
Let us denote by $B^{i}_{\frac{1}{\sqrt{N}}}=\left[ x_{i}-\frac{1}{\sqrt{N}}, x_{i}+\frac{1}{\sqrt{N}}\right]$, $i=1,\ldots,m$, some neighborhoods of the equilibria. A second order expansion of the entropy function around its minima $x_{i}$ yields
$$I(\tfrac{j}{N})=I(x_i) + I'(x_i)(\tfrac{j}{N}-x_i) + \frac{I''(\psi_j)}{2}(\tfrac{j}{N}-x_i)^2 = \frac{I''(\psi_j)}{2}(\tfrac{j}{N}-x_i)^2, $$
for all $\frac{j}{N}\in B^{i}_{\frac{1}{\sqrt{N}}}$ and with $\psi_j$ such that $|x_i-\psi_j | \leq |\frac{j}{N}-\psi_j |$. Notice that for all $\frac{j}{N}\in B^{i}_{h}$, $(\tfrac{j}{N}-x_i)^2\leq \frac{1}{N}$. This consideration, as well as Lemma~\ref{lemma:pi(j/N)} show that
\begin{align*}
\bar\pi_N\lp B^{i}_{\frac{1}{\sqrt{N}}}\rp &= \sum_{j:\frac{j}{N}\in B^{i}_{\frac{1}{\sqrt{N}}}} \bar\pi_N\lp \tfrac{j}{N}\rp \geq \sum_{j:\frac{j}{N}\in B^{i}_{\frac{1}{\sqrt{N}}}}\underline\gamma \,\frac{\bar\pi_N(0)}{\sqrt{N}e^{N J_0}}\,\frac{b(0)}{b(\frac{j}{N})}\,\frac{1}{\sqrt{r(\frac{j}{N})}}\,e^{-N I(\frac{j}{N})}\\
&\geq \underline\gamma \,b(0) \,\frac{\bar\pi_N(0)}{e^{N J_0}}\, \frac{1}{\sqrt{N}}\sum_{j:\frac{j}{N}\in B^{i}_{\frac{1}{\sqrt{N}}}}\frac{1}{\sqrt{b(\frac{j}{N})d(\frac{j}{N})}}\,e^{- \frac{I''(\psi_j)}{2}}.
\end{align*}
All terms $I''(\psi_j)$ can be upper bounded by a constant, as $\psi_j$ belongs to $B^{i}_{\frac{1}{\sqrt{N}}}$. Furthermore, there are approximatively $2\sqrt{N}$ terms in the sum $\sum_{j:\frac{j}{N}\in B^{i}_{\frac{1}{\sqrt{N}}}}$. Thus, applying Lemma~\ref{lemma:pi(0):J(0)}, we find that 
\begin{align*}
\bar\pi_N\lp B^{i}_{\frac{1}{\sqrt{N}}}\rp &\geq \underline\gamma \,b(0) \,\frac{\bar\pi_N(0)}{e^{N J_0}}\, \inf_{j:\frac{j}{N}\in B^{i}_{\frac{1}{\sqrt{N}}}}\frac{1}{\sqrt{b(\frac{j}{N})d(\frac{j}{N})}}\\
&=\tilde\gamma_1 \,\frac{\bar\pi_N(0)}{e^{N J_0}}\xrightarrow[N\to\infty]{} \tilde\gamma_2 >0
\end{align*}
for two positive constants $\gamma_1$ and $\gamma_2$, where we used that $\frac{1}{\sqrt{b(x)d(x)}}$ can be lower bounded from 0 in the neighborhood of all $x_i$, $i=1,\ldots,m$. This shows that the steady measure $\bar\pi_N$ concentrates around the equilibria. 
\end{proof}

\begin{corollary}\label{cor:WeakConvergence}
Assume the same hypotheses as in Lemma \ref{lemma:pi:charges:equilibr}, then
$$\bar\pi_N \Rightarrow \sum_{i=1}^m c_i \delta_{x_i},$$
where $m$ is the number of roots $x_i$ of $I(x)$, $c_i>0$ are constants for all $i=1,...,m$ and where $\Rightarrow$ denotes the weak convergence of measures.
\end{corollary}
\proof
Let $B$ an open set in $\left[0,1 \right]$. Following Lemma~\ref{lemma:pi:charges:equilibr} it exists $\epsilon_i$ sufficiently small such that
$$\liminf_{N\to\infty}\bar\pi_N(B) \geq \liminf_{N\to\infty}\bar\pi_N\lp \bigcup_{i=1}^m \left[x_i - \epsilon_i \, x_i + \epsilon_i \right] \cap B \rp,$$
and where the union is upon disjoint subsets of $\left[0,\; 1 \right]$. It follows that
\begin{eqnarray*}
\liminf_{N\to\infty}\bar\pi_N(B) & \geq & \sum_{i=1}^m \liminf_{N\to\infty} \bar\pi_N \lp \left[x_i - \epsilon_i \, x_i + \epsilon_i \right] \cap B \rp \\
& = & \sum_{i=1}^m \bar\pi_N \lp \left[x_i - \epsilon_i \, x_i + \epsilon_i \right] \cap B \, | \, x_i \in B \rp \Prob(x_i \in B) + \\ 
& \quad & \bar\pi_N \lp \left[x_i - \epsilon_i \, x_i + \epsilon_i \right] \cap B \, | \, x_i \notin B \rp \Prob(x_i \notin B) \\
& \geq & \sum_{i=1}^m \bar\pi_N \lp \left[x_i - \epsilon_i \, x_i + \epsilon_i \right] \cap B \, | \, x_i \in B \rp \Prob(x_i \in B)
\end{eqnarray*}
We can choose $\epsilon_i$ sufficiently small such that, by Lemma~\ref{lemma:pi:charges:equilibr},
$$\bar\pi_N \lp \left[x_i - \epsilon_i \, x_i + \epsilon_i \right] \cap B \, | \, x_i \in B \rp \geq \bar\pi_N \lp \left[x_i - \epsilon_i \, x_i + \epsilon_i \right]\rp \geq c_i. $$ 
Thus,
$$\liminf_{N\to\infty}\bar\pi_N(B) \geq \sum_{i=1}^M c_i \delta_{x_i}(B),$$
which implies the weak convergence of the probability measures (see e.g.~\cite[Thm.~13.6]{Klenke2008Book}).
\endproof

\begin{lemma}\label{lemma:estimates:e1}
Suppose Assumptions~\ref{assump:r} are satisfied and let 
$$\mathcal{E}^1(\tfrac{j}{N})=\sum_{k=1}^{j}\ln(r(\tfrac{k}{N}))-N\int_{\frac{1}{N}}^{\frac{j}{N}}\ln(r(u))\,{\rm d}u\\$$
denote a error term. There exist constants $\gamma_{1}, \gamma_{2}, \gamma_{3}, \gamma_{4}$ such that for $N$ large enough, the following inequalities
\begin{align*}
\gamma_{1}+\ln\lp\tfrac{1}{N}\rp \leq \,&N\int_{0_+}^{\frac{1}{N}} \ln\lp r(u)\rp\,{\rm d}u \,\leq \gamma_{2}+\ln\lp\tfrac{1}{N}\rp\\
\gamma_{1}+\ln\lp\tfrac{1}{N}\rp \leq \,&N\int_{1-\frac{1}{N}}^{1_-} \ln\lp r(u)\rp\,{\rm d}u \,\leq \gamma_{2}+\ln\lp\tfrac{1}{N}\rp
\end{align*}
and
$$\gamma_{3}+\frac{1}{2}\,\ln\lp\tfrac{1}{N}\rp+\frac{1}{2}\,\ln( r(\tfrac{j}{N})) \leq \mathcal{E}^1(\tfrac{j}{N}) \leq \gamma_{4}+\frac{1}{2}\,\ln\lp\tfrac{1}{N}\rp+\frac{1}{2}\,\ln( r(\tfrac{j}{N}))$$
are satisfied for all $1\leq j< 1$.
\end{lemma}
\begin{proof}
The first two pairs of inequalities follows from the definition of $r$ and from~\eqref{eq:condition:bx:dx}, since, for $N$ large, we have
$$N\int_{0_+}^{\frac{1}{N}} \ln\lp r(u)\rp\,{\rm d}u\approx N\int_{0_+}^{\frac{1}{N}} \ln\lp u\rp\,{\rm d}u = \ln\lp\tfrac{1}{N}\rp -1$$
and similarly for the integral from $1-\frac{1}{N}$ to $1_-$. In order to prove the last pair of inequalities, we use the trapezoidal rule in numerical analysis which states that, for any $\mathcal{C}^{2}$ function $g$ in an interval $[a,b]$, there exists $\xi\in[a,b]$ such that 
$$\frac{1}{2}\,g(a)+\frac{1}{2}\,g(b)-\frac{1}{b-a}\,\int_{a}^{b}g(u)\,{\rm d}u=\frac{(b-a)^{2}}{12}\,g^{\prime\prime}(\xi).$$
This implies that for any $N\geq1$ and any $1\leq j\leq N$, there exist $(\xi_{k})_{k=1,\ldots,j-1}$ with $\xi_{k}\in[\frac{k}{N},\frac{k+1}{N}]$ such that 
$$\sum_{k=1}^{j} g(\tfrac{k}{N})-N\int_{\frac{1}{N}}^{\frac{j}{N}}g(u)\,{\rm d}u = \frac{1}{12 N^{2}}\sum_{k=1}^{j-1} g^{\prime\prime}(\xi_{k}) +\frac{1}{2}\,g(\tfrac{1}{N}) +\frac{1}{2}\,g(\tfrac{j}{N}).$$
By assumption, $\ln (r)$ is only piecewise ${\cal C}^2$ on $[\frac{1}{N}\,,1-\frac{1}{N}]$. Without generality, we may assume that for all $i$, $d_i\in\Q$ so that for $N$ large enough, $d_i N\in\N$. The case $d_i\in\R$ is obtained similarly, using a density argument. Let $K(j)$ be the number of $d_i$ smaller than $\frac{j}{N}$. Then, for all $1\leq j< N$, 
\begin{align*}
\mathcal{E}^1(\tfrac{j}{N}) =& \sum_{i=1}^{j-1}\lac \frac{1}{2}\ln(r(\tfrac{k}{N}+))+ \frac{1}{2}\ln(r(\tfrac{k+1}{N}-))-N\int_{\frac{k}{N}}^{\frac{k+1}{N}}\ln(r(u))\,{\rm d}u \rac \\
& +\ln(r(\tfrac{1}{N}))-\frac{1}{2}\ln(r(\tfrac{1}{N}+)) +\ln(r(\tfrac{j}{N}))-\frac{1}{2}\ln(r(\tfrac{j}{N}-)) \\
& + \sum_{k=2}^{j-1} \lac\ln(r(\tfrac{k}{N}))-\frac{1}{2}\ln(r(\tfrac{k}{N}-))-\frac{1}{2}\ln(r(\tfrac{k}{N}+))\rac\\
\approx& \frac{1}{12 N^{2}} \sum_{k=1}^{j-1} (\ln(r))^{\prime\prime}(\xi_{k})+ \,\frac{1}{2}\ln(r(\tfrac{1}{N}))+\frac{1}{2}\ln(r(\tfrac{j}{N})) \\
& + \sum_{i=1}^{K(j)}\lp \frac{1}{2}\ln(r(d_{i}-))- \frac{1}{2}\ln(r(d_{i}+)) \rp ,
\end{align*}
with $\xi_{k}\in[\frac{k}{N},\frac{k+1}{N}]$, where we used that $r$ is left-continuous in $]0,1[$, piecewise ${\cal C}^2$ and that for $N$ large enough, $r$ is continuous at $\frac{1}{N}$. The last sum of the above expression is bounded. Let us now consider the first term. Near $0$, $(\ln(r))^{\prime\prime} (\frac{k}{N})$ is approximatively equal to $-\frac{N^2}{k^2}$ (similarly near $1$, it is approximatively equal to $\frac{N^2}{N^2-k^2}$). Recalling that $\sum_{k\geq1}\frac{1}{k^2}$ converges and that away from $0$ and $1$, $(\ln(r))^{\prime\prime} (\frac{k}{N})$ is bounded, we deduce that there exist constants which lower (resp.~upper) bound $\frac{1}{12 N^{2}} \sum_{k=1}^{j-1} (\ln(r))^{\prime\prime}(\xi_{k})$. This concludes the proof.
\end{proof}

\begin{lemma}\label{lemma:pi(j/N)}
Suppose condition \eqref{Condition1} and Assumptions~\ref{assump:r} are satisfied. Then, there exist two positive constants $\underline\gamma, \bar\gamma$, such that 
$$\underline\gamma \,\frac{\bar\pi_N(0)}{\sqrt{N}e^{N J_0}}\,\frac{b(0)}{b(\frac{j}{N})}\,\frac{1}{\sqrt{r(\frac{j}{N})}}\,e^{-N I(\frac{j}{N})} \leq
\bar\pi_N\lp \tfrac{j}{N}\rp \leq
\bar\gamma \,\frac{\bar\pi_N(0)}{\sqrt{N}e^{N J_0}}\,\frac{b(0)}{b(\frac{j}{N})}\,\frac{1}{\sqrt{r(\frac{j}{N})}}\,e^{-N I(\frac{j}{N})} $$
for all $1\leq j\leq N-1$.
\end{lemma}
\begin{proof}
For all $j\geq1$, by~\eqref{Steady}, by the convergence of $b^{(N)}$ and $d^{(N)}$, and by~\eqref{NewCondition}, there exists a constant $\gamma$ such that for $N$ large enough, we have that
\begin{align*}
\bar\pi_N\lp \frac{j}{N}\rp
&=\bar\pi_N(0)\,\frac{b^{(N)}(0)}{b^{(N)}(\frac{j}{N})}\,\prod_{k=1}^{j}\frac{b^{(N)}(\frac{k}{N})}{d^{(N)}(\frac{k}{N})} =\bar\pi_N(0)\,\frac{b^{(N)}(0)}{b^{(N)}(\frac{j}{N})}\,\exp\lac-\sum_{k=1}^{j} \ln\lp \frac{d^{(N)}(\frac{k}{N})}{b^{(N)}(\frac{k}{N})}\rp\rac\\
&\geq \bar\pi_N(0)\,\frac{b(0)}{b(\frac{j}{N})}\,\exp\lac-\sum_{k=1}^{j} \ln r(\tfrac{k}{N})\rac e^{-\gamma}.
\end{align*}
By definition of $J(x)$ and of $I(x)$ in~\eqref{FreeEnergy}, we find, applying Lemma~\ref{lemma:estimates:e1}, that 
\begin{eqnarray*}
\bar\pi_N( \tfrac{j}{N}) 
 & \geq& \bar\pi_N(0)\,\frac{b(0)}{b(\frac{j}{N})}\,\exp\lac- N J(\tfrac{j}{N})+N\int_{0_+}^{\frac{1}{N}} \ln\lp r(u)\rp\,{\rm d}u- \mathcal{E}^1(\tfrac{j}{N})\rac e^{-\gamma}\\
 &\geq& \frac{\bar\pi_N(0)}{e^{N J_0}} \,\frac{b(0)}{b(\frac{j}{N})}\,\exp\lac- N I(\tfrac{j}{N})-\gamma_1 +\ln\lp\tfrac{1}{N}\rp-\gamma_4 -\frac{1}{2}\ln\lp\tfrac{1}{N}\rp-\frac{1}{2}\ln( r(\tfrac{j}{N}))\rac e^{-\gamma}\\
 &= &\frac{\bar\pi_N(0)}{e^{N J_0}} \,\frac{b(0)}{b(\frac{j}{N})}\,\exp\lac- N I(\tfrac{j}{N})\rac e^{-\gamma-\gamma_1-\gamma_4} \frac{1}{\sqrt{N}\sqrt{r(\frac{j}{N})}}.
\end{eqnarray*}
Setting $\underline\gamma=e^{-\gamma-\gamma_1-\gamma_4}$, we get the desired lower bound. We proceed analogously for the upper bound by setting $\bar\gamma=e^{\gamma-\gamma_2-\gamma_3}$.
\end{proof}

From~\cite{chan1998largedeviation}, we already know that $\lim_{N\to\infty}\frac{1}{N}\ln(\bar\pi_N(0))=J_0$. Here, we need a slightly stronger result.

\begin{lemma}\label{lemma:pi(0):J(0)}
Suppose condition \eqref{Condition1} and Assumptions~\ref{assump:r} are satisfied. Then,
$$\lim_{N\to\infty} \frac{\bar\pi_N(0)}{e^{N J_0}}=\gamma_5 < +\infty.$$
\end{lemma}
\begin{proof}
By~\eqref{Steady:at:0} and by Lemma~\ref{lemma:pi(j/N)}, we have that
\begin{align*}
(\bar\pi_N(0))^{-1}&= 1+ \frac{\bar\pi_N(1)}{\bar\pi_N(0)}+\sum_{k=1}^{N-1}\frac{\bar\pi_N(\frac{k}{N})}{\bar\pi_N(0)}     \leq       1+ \frac{\bar\pi_N(1)}{\bar\pi_N(0)} + \bar\gamma \sum_{k=1}^{N-1}\frac{1}{\sqrt{N}}\,\frac{b(0)}{b(\frac{k}{N})}\,\frac{1}{\sqrt{r(\frac{k}{N})}}\,e^{-N J(\frac{k}{N})} \\
&= 1+ \frac{\bar\pi_N(1)}{\bar\pi_N(0)} + \bar\gamma \sqrt{N}\int_0^1 h(u)\,{\rm d}u + \bar\gamma \,\mathcal{E}^2_N
\end{align*}
where $h(x)=\frac{b(0)}{\sqrt{b(x)d(x)}}\,e^{-N J(x)}$ and $\mathcal{E}^2_N = \frac{1}{\sqrt{N}}\sum_{k=1}^{N-1} h(\tfrac{k}{N}) - \sqrt{N}\int_0^1 h(u)\,{\rm d}u$.

Firstly, proceeding similarly as in the proof of Lemma~\ref{lemma:pi(j/N)}, we have 
\begin{align*}
\frac{\bar\pi_N(1)}{\bar\pi_N(0)}&= \prod_{k=1}^{N-1}\frac{b^{(N)}(\frac{k}{N})}{d^{(N)}(\frac{k+1}{N})} =\frac{b^{(N)}(0)}{d^{(N)}(1)}\prod_{k=1}^{N-1}\frac{b^{(N)}(\frac{k}{N})}{d^{(N)}(\frac{k}{N})} \leq \frac{b(0)}{d(1)}\,\exp\lac-\sum_{k=1}^{N-1} \ln r(\tfrac{k}{N})\rac e^{\gamma}\\
& \leq \bar\gamma e^{-N J_0} \,\frac{b(0)}{d(1)}\,\exp\lac- N I(1-\tfrac{1}{N}) +\frac{1}{2}\,\ln\lp\tfrac{1}{N}\rp-\frac{1}{2}\,\ln( r(1-\tfrac{1}{N}))\rac\\
& \leq \bar\gamma e^{-N J_0} \,\frac{b(0)}{d(1)}\,\exp\lac- N I(1-\tfrac{1}{N}) +\frac{1}{2}\,\ln\lp\tfrac{1}{N}\rp-\frac{1}{2}\,\ln( N-1)\rac\\
& \leq \bar\gamma e^{-N J_0} \,\frac{b(0)}{d(1)}\,\frac{1}{N},
\end{align*}
since $I(x)\geq0$ and by assumptions on the behavior of $r(x)$ for $x$ near $1$. Secondly, consider the integral 
$$\sqrt{N}\int_0^1 h(u)\,{\rm d}u =e^{-N J_0} \sqrt{N}\int_0^1 \frac{b(0)}{\sqrt{b(u)d(u)}}\,e^{-N I(u)}\,{\rm d}u.$$ 
By Assumptions~\ref{assump:r}, for $x\approx0$, we have that $\frac{1}{\sqrt{b(x)d(x)}}\approx \frac{1}{\sqrt{(1-x)x}} \frac{1}{\sqrt{\ell_{b,0}\ell_{d,0}}}$ (and similarly for $x\approx1$). Since $r(x)=\frac{d(x)}{b(x)}$ is supposed to be piecewise $\mathcal{C}^2$ on $]0,1[$, we deduce that 
$$\int_0^1 \frac{b(0)}{\sqrt{b(u)d(u)}}{\rm d}u<+\infty.$$
Set $B_{\frac{1}{\sqrt{N}}}=\cup_{i=1}^{m}[ x_{i}-\tfrac{1}{\sqrt{N}}\,,\,x_{i}+\tfrac{1}{\sqrt{N}}]$. The volume of this subset of $[0,1]$ equals $\frac{2m}{\sqrt{N}}$ and by Assumptions~\ref{assump:r}, there exists $\delta>0$ such that 
$$I(x)>\delta, \quad \text{ for all } x\in[0,1]\setminus B_{\frac{1}{\sqrt{N}}}.$$
Therefore, 
\begin{align*}
 \sqrt{N}\int_0^1 \frac{b(0)}{\sqrt{b(u)d(u)}}\,e^{-N I(u)}\,{\rm d}u &= \sqrt{N}\int_0^1 \frac{b(0)}{\sqrt{b(u)d(u)}}\,e^{-N I(u)} \ind_{B_{\frac{1}{\sqrt{N}}}}(u)\,{\rm d}u \\
 &\quad + \sqrt{N}\int_0^1 \frac{b(0)}{\sqrt{b(u)d(u)}}\,e^{-N I(u)} \ind_{[0,1]\setminus B_{\frac{1}{\sqrt{N}}}}(u) \,{\rm d}u \\
 &\leq \sqrt{N} \frac{2m}{\sqrt{N}}\,\sup_{x\in B_{\frac{1}{\sqrt{N}}}} \frac{b(0)}{\sqrt{b(x)d(x)}} \\
 &\quad + \sqrt{N} e^{-N\delta}\int_0^1 \frac{b(0)}{\sqrt{b(u)d(u)}} \,{\rm d}u, 
\end{align*}
which goes towards a constant $c<+\infty$ when $N\to\infty$. Lastly, the error term $\mathcal{E}^2_N$ can be handled with the trapezoidal rule as follows (see proof of Lemma~\ref{lemma:estimates:e1})
\begin{align*}
\mathcal{E}^2_N &= \frac{1}{\sqrt{N}}\lp \sum_{k=1}^{N-1} h(\tfrac{k}{N}) - \sqrt{N}\int_{\frac{1}{N}}^{1-\frac{1}{N}} h(u)\,{\rm d}u\rp - \sqrt{N}\int_{0}^{\frac{1}{N}} h(u)\,{\rm d}u -N\int_{1-\frac{1}{N}}^{1} h(u)\,{\rm d}u\\
&= \frac{1}{\sqrt{N}}\lp \frac{1}{2}h(\tfrac{1}{N}) + \frac{1}{2}h(1-\tfrac{1}{N}) + \frac{1}{2}\sum_{i=1}^{K}\lp \ln(h(d_{i}-))-\ln(h(d_{i}+)) \rp+\frac{1}{12N^{2}}\sum_{k=1}^{N-2} h^{\prime\prime}(\xi_{k}) \rp\\
&\quad - \sqrt{N}\int_{0}^{\frac{1}{N}} h(u)\,{\rm d}u -\sqrt{N}\int_{1-\frac{1}{N}}^{1} h(u)\,{\rm d}u,
\end{align*}
with $\frac{k}{N}\leq\xi_{k}\leq\frac{k+1}{N}$. By assumption, on one hand we have that
$$ \tfrac{1}{\sqrt{N}} \tfrac{1}{2}h(\tfrac{1}{N}) \approx \tfrac{1}{\sqrt{N}} \tfrac{1}{2}h(1-\tfrac{1}{N})\approx \tfrac{1}{2} e^{\ln(N)+1}=\tfrac{e}{2} N$$
and the sum over $d_i$, the discontinuities of $r(x)$, is bounded by a constant $c_2$. On the other hand, the terms $h^{\prime\prime}(\xi_{k})$ can be bounded in any closed subinterval of $]0,1[$. Thus, the second derivative has to be controlled only near $0$ and $1$. Since at these boundaries $h(x)\propto (x(1-x))^{-\frac{1}{2}}e^{-NJ(x)}$, since $J(x)\approx\int_{0_+}^{x}\ln(u)\,{\rm d}u$ for $x\approx0$ and $J(x)\approx\int_{0_+}^{1_-}\ln(r(u))\,{\rm d}u-\int_{x}^{1_-}\ln(r(1-u))\,{\rm d}u$ for $x\approx1$, we can show by computing the first and second derivatives of these approximations that 
$$h^{\prime\prime}(\tfrac{1}{N})=N^{\frac{7}{2}}\lp\ln(N)\rp^2+o(N^4)=h^{\prime\prime}(1-\tfrac{1}{N}).$$
Therefore, 
$$ \frac{1}{\sqrt{N}}\,\frac{1}{12N^{2}}\sum_{k=1}^{N-2} h^{\prime\prime}(\xi_{k}) \leq \frac{1}{12N^{\frac{3}{2}}}\max_{1\leq k\leq N-2} h^{\prime\prime}(\xi_{k}) =N^2\lp\ln(N)\rp^2+o(N^{\frac{5}{2}}).$$
Finally, 
\begin{equation}\label{eq:two:integrals:near:0:1}
 - \sqrt{N}\int_{0}^{\frac{1}{N}} h(u)\,{\rm d}u \leq 0\quad \text{ and } -\sqrt{N}\int_{1-\frac{1}{N}}^{1} h(u)\,{\rm d}u \leq 0.
\end{equation}
Bringing every estimates together, we obtain a (loose) upper-bound for $(\bar\pi_N(0))^{-1}$,
\begin{align*}
(\bar\pi_N(0))^{-1}&= 1+ \frac{\bar\pi_N(1)}{\bar\pi_N(0)} + \bar\gamma \sqrt{N}\int_0^1 h(u)\,{\rm d}u + \bar\gamma \,\mathcal{E}^2_N\\
&\leq 1+\bar\gamma\lp e^{-N J_0} \,\frac{b(0)}{d(1)} \,\frac{1}{N}+ c_1+ \frac{e}{2} N+\frac{e}{2} N+\frac{1}{N}c_2+ N^2\lp\ln(N)\rp^2+o(N^{\frac{5}{2}})\rp\\
&\leq\tilde\gamma e^{-N J_0} +o(N^{3}),
\end{align*}
with $\tilde\gamma>0$. Thus, 
$$\bar\pi_N(0)\geq \tilde\gamma e^{N J_0} \frac{1}{1+o(N^{3})e^{N J_0}}=\tilde\gamma e^{N J_0}o(1)$$
and we can proceed analogously to find an upper bound of the same order. The only difference arises for the integrals in~\eqref{eq:two:integrals:near:0:1} which can be lower bounded by a term of order at most $-N^2$. This concludes the proof.
\end{proof}

\section{Proof of the results from Section~\ref{s.allosteric}}
\subsection{Proof of Proposition~\ref{prop:mixture0}}
In proposition~\ref{prop:mixture0}, $\bar\pi(k)$ is written as a mixture of two binomial distribution, 
\begin{equation*}
\bar\pi(k) =\frac{{N\choose k}(\frac{v}{1+v})^k (\frac{1}{1+v})^{N-k}}{1+(\frac{1+\varepsilon v}{1+v})^N L_1}
      +\frac{{N\choose k}(\frac{\varepsilon v}{1+\varepsilon v})^k (\frac{1}{1+\varepsilon v})^{N-k}}{1+(\frac{1+ v}{1+\varepsilon v})^N L_1^{-1}}.
\end{equation*}

\begin{proof}[Proof of Proposition~\ref{prop:mixture0}]
The stationary probability measure $\mu_N$ of the Markov chain $(N(t),W(t))$ can be defined through positive weights $w_N(k,A)$ and $w_N(k,I)$ by setting
$$\mu_N(k,A)=\frac{w_N(k,A)}{Z_N} \qquad\text{ and }\qquad\mu_N(k,I)=\frac{w_N(k,I)}{Z_N},$$
where $Z_N = \sum_{0\le k\le N}(w_N(k,A)+w_N(k,I))$ is the normalization constant.
The transitions of the upper layer of the strip correspond to a classical Ehrenfest urn process, and it is then natural to set 
$$w_N(k,A)={N\choose k}v^k,$$
with a similar expression for the chain restricted to the lower layer,
$$w_N(k,I)= L_1{N\choose k}(v\varepsilon)^k.$$
For the active state, notice that the balance equation is satisfied
\begin{align*}
&w_N(k-1,A)q_N((k-1,A),(k,A))+w_N(k+1,A)q_N((k+1,A),(k,A)) + w_N(k,I)q_N((k,I),(k,A))\\
&= w_N(k,A)\big[q_N((k,A),(k+1,A))+q_N((k,A),(k-1,A))+q_N((k,A),(k,I))\big],
\end{align*}
and similarly for the inactive state, showing that $\mu_N$ is the steady state distribution of the Markov chain for any $L_1 >0$. In conclusion, 
$$\bar\pi(k)=\mu_N(k,A) +\mu_N(k,I)=\frac{{N\choose k}v^k}{(1+v)^N +L_1 (1+\varepsilon v)^N} + \frac{L_1{N\choose k}(v\varepsilon)^k}{(1+v)^N +L_1 (1+\varepsilon v)^N}.$$
\end{proof}

\subsection{Proof of Proposition~\ref{asymptoticallostericHill}}

Proposition~\ref{prop:mixture0} shows that the steady state distribution is a mixture of two binomial distributions $\pi_1=\mathcal{B}(N,\frac{\varepsilon v}{1+\varepsilon v})$ and $\pi_1=\mathcal{B}(N,\frac{v}{1+v})$. Suppose for the sequel that $L_1 = L_1(N)=\varepsilon^{-N/2}$, as in \cite{ryerson2014ultrasensitivity,EncisoKelloggVargas2014}, so that the coefficient of the mixture becomes
\begin{equation}\label{eq:encis:mixt:coeff}
\alpha(v)=\frac{1}{1+\lp\frac{\sqrt{\varepsilon}(1+ v)}{1+\varepsilon v}\rp^{N}}\,.
\end{equation}
Using Lemma~\ref{Mixture} we can find the asymptotic related Hill coefficient and show that, for the critical concentration $v_c$, the system is ultrasensitive.

\begin{proof}[Proof of Proposition~\ref{asymptoticallostericHill}]
Let $v\neq v_c$ and, as in Example~\ref{Adair2}, use the formula in Lemma~\ref{Mixture} with a second order expansion of $a(x)$ in oder to asymptotically estimate the covariances.  The expansion is taken at $p=\varepsilon v/(1+\varepsilon v)$ when considering distribution $\pi_1$ and respectively at $p=v/(1+v)$ when considering $\pi_2$. This gives
\begin{eqnarray*}
N {\rm Cov}_{\pi_1}\lp a( \tfrac{|n|}{N}), \tfrac{|n|}{N} \rp & \xrightarrow[N\to\infty]{} & a'\lp \tfrac{\varepsilon v}{1+\varepsilon v} \rp \frac{\varepsilon v}{1+\varepsilon v} \frac{1}{1+\varepsilon v}\\
N {\rm Cov}_{\pi_2}\lp a( \tfrac{|n|}{N}), \tfrac{|n|}{N} \rp & \xrightarrow[N\to\infty]{} & a'\lp \tfrac{v}{1+v} \rp \frac{v}{1+v} \frac{1}{1+v} 
\end{eqnarray*}
Furthermore in our settings, $f(\infty) = a(1)$ and
$$f(v) = \left< a (\tfrac{|n|}{N}) \right>_{\alpha(v) \pi_1 + (1-\alpha(v))\pi_2} = \alpha(v)a( \tfrac{\varepsilon v}{1+\varepsilon v})  + (1-\alpha(v))a( \tfrac{v}{1+v}).$$
The derivative of $\alpha$ is given by
$$\alpha'(v) = - \frac{(1-\varepsilon)\varepsilon^{\frac{N}{2}} N \lp \tfrac{1+\varepsilon v}{1 + v} \rp^{N-1}}{ (1+v)^2 \lp \varepsilon^{\frac{N}{2}} + \lp \tfrac{1+\varepsilon v}{1+v} \rp^N \rp^2}.$$ 
Putting all together in the formula in Lemma~\ref{Mixture} and taking the limit gives
\begin{eqnarray}\label{lemma:formulaHillEnciso:proof}
\lim_{N\to\infty} \eta_H(v) = \lim_{N\to\infty} a(1) \lp \varepsilon^{\frac{N}{2}} + \lp \tfrac{1+\varepsilon v}{1 + v} \rp^N \rp^2  \frac{A(N,\varepsilon,v) + B(N,\varepsilon,v) + C(N,\varepsilon,v)}{D(N,\varepsilon,v)},
\end{eqnarray}
where we define for convenience
\begin{eqnarray*}
A(N,\varepsilon,v) & = & - v\, \frac{(1-\varepsilon)\varepsilon^{\frac{N}{2}}\,N \lp \tfrac{1+\varepsilon v}{1 + v} \rp^{N} \lp a(\tfrac{\varepsilon v}{1+\varepsilon v}) -a(\tfrac{v}{1+v}) \rp}{(1+v)(1+\varepsilon v)\lp \varepsilon^{\frac{N}{2}} + \lp \tfrac{1+\varepsilon v}{1 + v} \rp^N \rp^2 } \\
B(N,\varepsilon,v) & = & \frac{v\, a'(\tfrac{v}{1+v})}{(1+v)^2 \lp 1 + \varepsilon^{-\frac{N}{2}} \, \lp \tfrac{1+\varepsilon v}{1 + v}  \rp^N \rp}\\
C(N,\varepsilon,v) & = & \frac{\varepsilon \, v\,  \varepsilon^{-\frac{N}{2}} \lp \tfrac{1+\varepsilon v}{1 + v} \rp^N \, a' (\tfrac{\varepsilon v}{1 + \varepsilon v})}{(1 + \varepsilon v)^2 \lp 1 + \varepsilon^{-\frac{N}{2}} \lp \tfrac{1+\varepsilon v}{1 + v} \rp^N \rp} \\
D(N,\varepsilon,v) & = & a(1)\varepsilon^N a(\tfrac{v}{1+v}) + a(1)\lp \tfrac{1+\varepsilon v}{1+v} \rp^N a(\tfrac{v}{1+v}) \varepsilon^{\frac{N}{2}} - \varepsilon^N \lp a(\tfrac{v}{1+v}) \rp^2 \\
& \quad & - \varepsilon^{\frac{N}{2}} \lp \tfrac{1+\varepsilon v}{1+v} \rp^N a(\tfrac{v}{1+v}) a(\tfrac{\varepsilon v}{1+\varepsilon v}) + a(1) \varepsilon^{\frac{N}{2}} \lp \tfrac{1+\varepsilon v}{1+v} \rp^N a(\tfrac{\varepsilon v}{1 + \varepsilon v}) \\
& \quad & + a(1) \lp \tfrac{1+\varepsilon v}{1 + v} \rp^{2N} a(\tfrac{\varepsilon v}{1+\varepsilon v}) - \varepsilon^{\frac{N}{2}} \lp \tfrac{1+\varepsilon v}{1 + v} \rp^N a(\tfrac{\varepsilon v}{1+\varepsilon v})a(\tfrac{ v}{1+ v}) \\
& \quad & - \lp \tfrac{1+\varepsilon v}{1+v} \rp^{2N} \lp a(\tfrac{\varepsilon v}{1 + \varepsilon v}) \rp^2.
\end{eqnarray*}
The main idea of the calculation will be to determine the asymptotic contribution of every term to the limit~\eqref{lemma:formulaHillEnciso:proof}. To perform this, an important step consists in proving that
\begin{eqnarray}\label{lemma:formulaHillEnciso:proof:step}
\lim_{N\to\infty} \varepsilon^{-\frac{N}{2}}\lp \frac{1+\varepsilon v}{1+v} \rp^{-N} D(N,\varepsilon,v) = \infty.
\end{eqnarray}
Indeed, after having multiplied every term of $D$ with the preceding factor, we obtain
\begin{eqnarray*}
\lim_{N\to\infty} \varepsilon^{-\frac{N}{2}}\lp \tfrac{1+\varepsilon v}{1+v} \rp^{-N} D(N,\varepsilon,v) & = & \quad \lim_{N\to\infty} \varepsilon^{\frac{N}{2}}\lp \tfrac{1+\varepsilon v}{1+v} \rp^{-N}a(\tfrac{v}{1+v})\lp a(1) - a(\tfrac{v}{1+v}) \rp \\
& \quad & + \lim_{N\to\infty} \varepsilon^{-\frac{N}{2}}\lp \tfrac{1+\varepsilon v}{1+v} \rp^{N}a(\tfrac{\varepsilon v}{1+\varepsilon v})\lp a(1) - a(\tfrac{\varepsilon v}{1+\varepsilon v}) \rp \\
& \quad & + a(1)a(\tfrac{v}{1+v}) - 2a(\tfrac{v}{1+v})a(\tfrac{\varepsilon v}{1+\varepsilon v}) + a(1)a(\tfrac{\varepsilon v}{1+\varepsilon v}).
\end{eqnarray*}
Assuming $\sqrt{\varepsilon}(1+v)>1+\varepsilon v$ and $\varepsilon < 1$ implies $\varepsilon^{\frac{N}{2}}\lp \tfrac{1+\varepsilon v}{1+v} \rp^{-N} \xrightarrow[N\to\infty]{} \infty$ and thus $\varepsilon^{-\frac{N}{2}}\lp \tfrac{1+\varepsilon v}{1+v} \rp^{N} \xrightarrow[N\to\infty]{} 0$. The assumption $\sqrt{\varepsilon}(1+v)<1+\varepsilon v$ reverses the previous conclusion, so that in every case and using the assumption of $a$ positive and strictly increasing on $\left[0,\, 1 \right]$ the limit~\eqref{lemma:formulaHillEnciso:proof:step} holds. 

Now the computation of~\eqref{lemma:formulaHillEnciso:proof} turns to three steps. Firstly, because~\eqref{lemma:formulaHillEnciso:proof:step} occurs at exponential speed, we compute
\begin{align*}
& \lim_{N\to\infty} a(1) \lp \varepsilon^{\frac{N}{2}} + \lp \tfrac{1+\varepsilon v}{1 + v} \rp^N \rp^2  \frac{A(N,\varepsilon,v)}{D(N,\varepsilon,v)}\\
& \quad =  \lim_{N\to\infty}-v\, \frac{a(1)(1-\varepsilon)N\lp a(\tfrac{\varepsilon v}{1+\varepsilon v})-a(\tfrac{v}{1+v}) \rp}{(1+v)(1+\varepsilon v)\varepsilon^{-\frac{N}{2}}\lp \frac{1+\varepsilon v}{1+v} \rp^{-N} D(N,\varepsilon,v)} = 0
\end{align*}

Secondly, noting that $\lp \varepsilon^{\frac{N}{2}} + \lp \tfrac{1+\varepsilon v}{1 + v} \rp^N \rp^2 = \varepsilon^N\lp 1 + \varepsilon^{-\frac{N}{2}}\lp \tfrac{1+\varepsilon v}{1 + v} \rp^N \rp^2$, we can write
\begin{align*}
& \lim_{N\to\infty} a(1) \lp \varepsilon^{\frac{N}{2}} + \lp \tfrac{1+\varepsilon v}{1 + v} \rp^N \rp^2  \frac{B(N,\varepsilon,v)}{D(N,\varepsilon,v)} \\
& \quad = \lim_{N\to\infty} \frac{a(1)\varepsilon^N v a'(\tfrac{v}{1+v})}{(1+v)^2 D(N,\varepsilon,v)} + \lim_{N\to\infty}\frac{a(1)\varepsilon^{\frac{N}{2}}\lp \tfrac{1+\varepsilon v}{1+v} \rp^Nva'(\tfrac{v}{1+v})}{(1+v)^2 D(N,\varepsilon,v)},
\end{align*}
where the second term is zero by~\eqref{lemma:formulaHillEnciso:proof:step} and the part containing $N$ in the first term can be computed as
\begin{eqnarray*}
\lim_{N\to\infty} \varepsilon^{-N}D(N,\varepsilon,v) & = & \quad \lim_{N\to\infty} \varepsilon^{-\frac{N}{2}} \lp \tfrac{1+\varepsilon v}{1+v} \rp^N \lp a(\tfrac{v}{1+v}) + a(\tfrac{\varepsilon v}{1+\varepsilon v}) \rp  \lp a(1) - a(\tfrac{\varepsilon v}{1+\varepsilon v}) \rp \\
& \quad & + \lim_{N\to\infty} \varepsilon^{-N} \lp \tfrac{1+\varepsilon v}{1+v} \rp^{2N} a(\tfrac{\varepsilon v}{1+\varepsilon v})\lp a(1) - a(\tfrac{\varepsilon v}{1+\varepsilon v}) \rp \\
& \quad & + \ a(\tfrac{v}{1+v})\lp a(1) - a(\tfrac{v}{1+v}) \rp.
\end{eqnarray*}
The asymptotic behavior of $\varepsilon^{-\frac{N}{2}} \lp \tfrac{1+\varepsilon v}{1+v} \rp^N$ depends on the values of $\varepsilon$ et $v$ such that
\begin{eqnarray*}
\lim_{N\to\infty} \frac{a(1)\varepsilon^N v a'(\tfrac{v}{1+v})}{(1+v)^2 D(N,\varepsilon,v)} = 
\begin{cases}
\frac{a(1)v a'\lp\frac{v}{1+v} \rp}{(1+v)^2 a\lp\frac{v}{1+v}\rp \lp a(1) - a\lp\frac{v}{1+v}\rp \rp}\,, & \text{ if } \sqrt{\varepsilon}(1+v)>1+\varepsilon v, \\
0, & \text{ otherwise}.
\end{cases}
\end{eqnarray*}
Finally, 
\begin{align*}
& \lim_{N\to\infty} a(1) \lp \varepsilon^{\frac{N}{2}} + \lp \tfrac{1+\varepsilon v}{1 + v} \rp^N \rp^2  \frac{C(N,\varepsilon,v)}{D(N,\varepsilon,v)} \\
& = \lim_{N\to\infty} \frac{a(1)\varepsilon v \varepsilon^{\frac{N}{2}} \lp \tfrac{1+\varepsilon v}{1+v} \rp^N a'(\tfrac{\varepsilon v}{1+\varepsilon v})}{(1+\varepsilon v)^2 D(N,\varepsilon,v)} + \lim_{N\to\infty} \frac{a(1)\varepsilon v \lp \tfrac{1+\varepsilon v}{1+v} \rp^{2N} a'(\tfrac{\varepsilon v}{1+\varepsilon v})}{(1+\varepsilon v)^2 D(N,\varepsilon,v)},
\end{align*}
where by~\eqref{lemma:formulaHillEnciso:proof:step} the first term goes to zero when $N$ is large and the second one can be handled analogously as above, noting that
\begin{eqnarray*}
\lim_{N\to\infty} \lp \tfrac{1+\varepsilon v}{1+v} \rp^{-2N} D(N,\varepsilon,v) & = & \quad \lim_{N\to\infty} \varepsilon^{\frac{N}{2}} \lp\tfrac{1+\varepsilon v}{1+\varepsilon v} \rp^{-N} a(\tfrac{v}{1+v}) \lp a(1) - a(\tfrac{v}{1+v})\rp \\
& \quad & + \lim_{N\to\infty} \varepsilon^{\frac{N}{2}} \lp\tfrac{1+ v}{1+ v} \rp^{-N} a(\tfrac{v}{1+v}) \lp a(1) - a(\tfrac{\varepsilon v}{1+\varepsilon v})\rp \\
& \quad & + \lim_{N\to\infty} \varepsilon^{N} \lp\tfrac{1+ v}{1+ v} \rp^{-2N} a(\tfrac{v}{1+v}) \lp a(1) - a(\tfrac{ v}{1+ v})\rp \\
& \quad & + a(\tfrac{\varepsilon v}{1+\varepsilon v})\lp a(1) - a(\tfrac{\varepsilon v}{1 + \varepsilon v})\rp
\end{eqnarray*}
so that
\begin{eqnarray*}
\frac{a(1)\varepsilon v \lp \tfrac{1+\varepsilon v}{1+v} \rp^{2N} a'(\tfrac{\varepsilon v}{1+\varepsilon v})}{(1+\varepsilon v)^2 D(N,\varepsilon,v)} \to
\begin{cases}
0, & \text{ if } \sqrt{\varepsilon}(1+v)>1+\varepsilon v, \\
\frac{a(1)\varepsilon v a'\lp\frac{\varepsilon v}{1+\varepsilon v} \rp}{(1+\varepsilon v)^2 a\lp\frac{\varepsilon v}{1+\varepsilon v}\rp \lp a(1) - a\lp\frac{\varepsilon v}{1+ \varepsilon v}\rp \rp}\,, & \text{ otherwise},
\end{cases}
\end{eqnarray*}
for $N\to\infty$. Putting all together leads to~\eqref{FormulaHillNonCrit}. Setting $v=v_c=1/\sqrt{\varepsilon}$ implies $\alpha(v_c) = \tfrac{1}{2}$ and $\alpha'(v_c) = \tfrac{N(\sqrt{\varepsilon}-\varepsilon)}{4(\sqrt{\varepsilon} + 1)}$ so that the part of $\eta_H(v)$ involving the covariances does not contribute to the limit. This leads then directly to formula~\eqref{FormulaHill} with 
$$C_{v_c} = \frac{(1-\varepsilon)}{(1+\sqrt{\varepsilon})(1+\frac{1}{\sqrt{\varepsilon}})} \,\frac{a(1)(a(\frac{1}{1+\sqrt{\varepsilon}})-a(\frac{\sqrt{\varepsilon}}{1+\sqrt{\varepsilon}}))}{(a(\frac{\sqrt{\varepsilon}}{1+\sqrt{\varepsilon}})+a(\frac{1}{1+ \sqrt{\varepsilon}}))
(2 a(1)-a(\frac{\sqrt{\varepsilon}}{1+\sqrt{\varepsilon}})-a(\frac{1}{1+ \sqrt{\varepsilon}}))}. $$
\end{proof}

\subsection{Proof of Lemma~\ref{allosteric:Ip}}
In this section we underline the link between Hill coefficient and the effective Hill coefficient for the special case of an allosteric phosphorylation process. As seen in (\ref{Mixture0}), the steady state $\bar\pi_N$ is a mixture of the binomial distributions $\pi_1=\mathcal{B}(N,\frac{\varepsilon v}{1+\varepsilon v})$ and $\pi_2=\mathcal{B}(N,\frac{v}{1+v})$, of coefficient
$$\alpha(v)=\frac{1}{1+\lp\frac{\sqrt{\varepsilon}(1+ v)}{1+\varepsilon v}\rp^{N}}.$$
Consider the quantile $v_q^{(N)}$ given by the equation $q=f(v_q^{(N)})$.

\begin{proof}[Proof of Lemma~\ref{allosteric:Ip}]
When $a(x)\equiv x$, one obtains that
$$f(v)=\alpha(v)\frac{\varepsilon v}{1+\varepsilon v}+(1-\alpha(v))\frac{v}{1+v} .$$
We first show that $v_q^{(N)}\longrightarrow v_c$ as $N\to\infty$. Suppose that $\liminf_{N\to\infty}v_q^{(N)} < v_c-\delta$, for some $\delta > 0$. The equation defining the quantile leads to
$$q(1+\Big(\frac{1+\varepsilon v_q^{(N)}}{\sqrt{\varepsilon}(1+v_q^{(N)})}\Big)^N)=\frac{v_q^{(N)}}{1+v_q^{(N)}}
+\Big(\frac{1+\varepsilon v_q^{(N)}}{\sqrt{\varepsilon}(1+v_q^{(N)})}\Big)^N \frac{\varepsilon v_q^{(N)}}{1+\varepsilon v_q^{(N)}}.$$
From assumption, $\forall N\ge 1$, $\exists m\ge N$ such that $v_q^{(m)}< v_c-\delta < v_c$, so that $(1+\varepsilon v_q^{(m)})/(\sqrt{\varepsilon}(1+v_q^{(m)})>1$. The above equations leads then to
$$q\sim \frac{\varepsilon v_q^{(m)}}{1+\varepsilon v_q^{(m)}}\hbox{ or } v_q^{(m)}\sim \frac{q}{\varepsilon(1-q)}.$$
But $1/2\leq q \leq 1/(1+\sqrt{\varepsilon})$, so that $q/(\varepsilon(1-q))> v_c$, which contradicts the inequality $v_q^{(m)}< v_c-\delta < v_c$. Hence $\liminf_{N\to\infty}v_q^{(N)}\geq v_c$. Similar arguments show that $\limsup_{N\to\infty}v_q^{(N)}\leq v_c$, proving that $\lim_{N\to\infty}v_q^{(N)}=v_c$, as required. Similar arguments show that $\lim_{N\to\infty}v_{1-q}^{(N)}=v_c$. 
\end{proof}

\section{Proof of Lemma~\ref{lemma:kaneko:largeDev}}

\begin{proof}
When $\gamma>1$, $\lim_{x\to0}\ln\lp r(x)\rp=-\infty$ and $\lim_{x\to1}\ln\lp r(x)\rp=\infty$. The continuity of $\ln(r(x))$ implies the existence of a root in the interval $\left]0,\,1\right[$. A necessary condition for the existence of at least three roots is that the function is not monotone increasing. This happens if the derivative 
$$\frac{\textrm{d}}{\textrm{d}x}\ln\lp r(x)\rp = -\,\frac{1+(x-1)x\ln\lp\gamma\rp}{(x-1)x}$$
is negative for some $x\in\left[0,\,1 \right]$. This holds true if and only if conditions $(C1)$ and $(C2)$ are satisfied. Let $x_{min}$ and $x_{max}$ be local minima and maxima of $\ln\lp r(x)\rp$; this function possesses at least two different roots when $\ln\lp r(x_{min})\rp<0$ and $\ln\lp r(x_{max})\rp>0$ which occurs when $(C3)$ is satisfied.
 
Let $x_1$, $x_2$ and $x_3$ be the zeros of $\ln\lp r(x) \rp$ and assume $(C4)$. Then $x_2=0.5$ and $\ln\lp r(x+0.5) \rp$ is odd. Hence,
$$J(x_3) = \int_0^{x_3}\ln\lp r(x)\rp\textrm{d}x = \int_0^{x_1}\ln\lp r(x)\rp\textrm{d}x +\underbrace{\int_{x_1}^{x_3}\ln\lp r(x)\rp\textrm{d}x}_{=0} = J(x_1),$$
so that $I$ possesses two roots.
\end{proof}

\section*{Acknowledgment}
We are grateful to our colleague Ioan Manolescu for helpful discussions on sharp-threshold results and influence functions.


\bibliographystyle{plain}    
\bibliography{arxiv}

\end{document}